\documentclass[journal]{IEEEtran}
\usepackage{subfigure,cite,graphicx,amsmath,amssymb,mathrsfs,epsf}
\usepackage{cite,graphicx,epsfig,subfigure,amsmath,amssymb,mathrsfs,epsf,graphics,color}
\usepackage{psfrag}

\newtheorem{theorem}{Theorem}%[section]x
\newtheorem{lemma}{Lemma}%[section]

\newtheorem{corollary}{Corollary}%[section]
\newtheorem{remark}{Remark}%[section]

\def \E{\operatorname{E}}
\def \P{\operatorname{Pr}}
\def \tr{\operatorname{tr}}

\begin{document}
\title{Capacity Scaling of Wireless Ad Hoc Networks: Shannon Meets Maxwell}
\author{Si-Hyeon~Lee,~\IEEEmembership{Student Member,~IEEE,} and~Sae-Young~Chung,~\IEEEmembership{Senior Member,~IEEE}
\thanks{This work was partially supported by the Defense Acquisition Program Administration and the Agency for Defense Development under the contract UD060048AD. The material in this paper was presented in part at the IEEE International Symposium on Information Theory, Toronto, Canada, July 2008 and at  the IEEE International Symposium on Information Theory, Austin, USA, June 2010.}
\thanks{S.-H. Lee and S.-Y. Chung are with the Department of Electrical Engineering, Korea Advanced Institute of Science and Technology (KAIST), Daejeon, South Korea (e-mail: sihyeon@kaist.ac.kr; sychung@ee.kaist.ac.kr).}
%\thanks{Copyright (c) 2011 IEEE. Personal use of this material is permitted. However, permission to use this material for any other purposes must be obtained from the IEEE by sending a request to pubs-permissions@ieee.org.}
}

\maketitle

\begin{abstract}
In this paper, we characterize the information-theoretic capacity scaling of wireless ad hoc networks with $n$ randomly distributed nodes. By using an exact channel model from Maxwell's equations, we successfully resolve the conflict in the literature between the linear capacity scaling by \"{O}zg\"{u}r et al. and the degrees of freedom limit given as the ratio of the network diameter and the wavelength $\lambda$ by Franceschetti et al. In dense networks where the network area is fixed, the capacity scaling is given as the minimum of $n$ and the degrees of freedom limit $\lambda^{-1}$ to within an arbitrarily small exponent. In extended networks where the network area is linear in $n$, the capacity scaling is given as the minimum of $n$ and the degrees of freedom limit $\sqrt{n}\lambda^{-1}$ to within an arbitrarily small exponent. Hence, we recover the linear capacity scaling by \"{O}zg\"{u}r et al. if $\lambda=O(n^{-1})$ in dense networks and if $\lambda=O(n^{-1/2})$ in extended networks. Otherwise, the capacity scaling is given as the degrees of freedom limit characterized by Franceschetti et al. For achievability, a modified hierarchical cooperation is proposed based on a lower bound on the capacity of multiple-input multiple-output channel between two node clusters using our channel model.
\end{abstract}

\allowdisplaybreaks
\section{Introduction}
\IEEEPARstart{P}{ioneered} by Gupta and Kumar in \cite{GuptaKumar:00}, the capacity scaling in wireless ad hoc networks has been actively studied over the last decade. In this research, we consider $n$ uniformly and independently distributed nodes in a unit area (dense network) or an area of $n$ (extended network), each of which wanting to communicate to a random destination at the same rate of $R(n)$. The goal is to find out the maximally achievable scaling of the aggregate throughput $T(n)=nR(n)$ with $n$. In their seminal paper \cite{GuptaKumar:00}, Gupta and Kumar showed that throughput scaling higher than $O(\sqrt{n})$ cannot be achieved if each node treats interference as noise and that the multihop scheme can achieve $\Theta(\sqrt{n}/\log n)$.\footnote{In this paper, we use the following asymptotic notations \cite{Knuth:76}: (i) $f(n)=O(g(n))$ if $f(n)\leq kg(n)$ as $n$ tends to infinity for some constant $k$. (ii) $f(n)=\Theta(g(n))$ if $k_1g(n)\leq f(n)\leq k_2g(n)$ as $n$ tends to infinity for some constants $k_1$ and $k_2$. (iii) $f(n)=\Omega(g(n))$ if $f(n)\geq kg(n)$ as $n$ tends to infinity for some constant $k$.}
This gap was closed in~\cite{Franceschetti:07}, where it was shown that the multihop via percolation theory can achieve $\Theta(\sqrt{n})$. To information theorists, a natural question is what the \emph{information-theoretic} capacity scaling is without such underlying physical-layer assumptions.

The information-theoretic capacity scaling is highly dependent on the channel model. Furthermore, it is important to use a realistic channel model to get results that are closer to reality. In wireless networks in line-of-sight (LOS) environments, where the spatial locations of nodes are fixed with sufficiently large inter-node separation compared to the wavelength, the baseband-equivalent channel response between two nodes $k$ and $i$ is given as
\begin{align}
\frac{\sqrt{G}}{d_{ik}}\exp\left(-j\frac{2\pi}{\lambda}d_{ik}\right) \label{eqn:maxwell}
\end{align}
from Maxwell's equations where $j=\sqrt{-1}$, $d_{ik}$ is the distance between nodes $k$ and $i$, $\lambda$ denotes the wavelength $\frac{c}{f_c}$ where $c$ is the speed of light and $f_c$ is the carrier frequency, and $G=\frac{\lambda^2 G_l}{16\pi^2}$  by Friis' formula where $G_l$ is the product of the transmit and receive antenna gains.

Recently, \"{O}zg\"{u}r et al. characterized the information-theoretic capacity scaling in \cite{Ozgur:07}. Instead of using the exact channel model (\ref{eqn:maxwell}) with a distance dependent phase, however, they assumed that the baseband-equivalent channel response between two nodes $k$ and $i$ is given as
\begin{align}
\frac{\sqrt{G}}{d_{ik}}\exp\left(j\theta_{ik}\right) \label{eqn:nomaxwell}
\end{align}
where $\theta_{ik}$ is independent and identically distributed (i.i.d.). For this channel model, the capacity scaling is shown to be arbitrarily close to linear in both dense and extended networks, which means that each source can communicate to its destination as if there were no interference. A key component to achieve such a scaling is the cooperative multiple-input multiple-output (MIMO) transmission between two node clusters whose sizes are comparable to that of the network. If the penalty to form such a virtual MIMO is negligible, the classical MIMO results~\cite{Foschini:96,Telatar:99} under the i.i.d. channel phase assumption make the linear throughput scaling possible. Such an overhead is indeed shown to be arbitrarily small by using hierarchical cooperation (HC). In the HC scheme, each cluster forms a virtual antenna array using MIMO transmissions between  small scale clusters inside it. Similarly, each small scale cluster forms a virtual antenna array by MIMO transmissions between even smaller clusters inside it. This builds up a hierarchy and a plain time division multiple access (TDMA) is performed at the bottom hierarchy.

The i.i.d. phase assumption in (\ref{eqn:nomaxwell}) makes the throughput analysis easier in \cite{Ozgur:07}, but such an artificial assumption can lead to results contradicting the physics. Recently, the linear capacity scaling in \cite{Ozgur:07} turned out to be contradictory to the physical limit on degrees of freedom (DoF) when $\lambda$ is not sufficiently small. In \cite{Franceschetti:09}, Franceschetti et al. showed, using Maxwell's equations, that DoF in extended networks is limited by the ratio of the network diameter $\sqrt{n}$ and $\lambda$. By rescaling the network size, the DoF limit becomes $\lambda^{-1}$ in dense networks. This is a fundamental limitation independent of power attenuation and fading models. Hence, the linear capacity scaling in \cite{Ozgur:07} is in fact not attainable for $\lambda=\Omega(n^{-1})$ and  $\lambda=\Omega(n^{-1/2})$ in dense and extended networks, respectively. The cause of such a conflict is the i.i.d. channel phase assumption in \cite{Ozgur:07} that ignores the channel correlation due to the distance dependent channel phase.

Two contradictory results \cite{Ozgur:07,Franceschetti:09} highlight the importance of exact channel models based on Maxwell's equations.  Thus, the ultimate goal would be the characterization of the information-theoretic capacity scaling of wireless ad hoc networks from Maxwell's equations without any artificial assumptions. In this paper, we accomplish this goal by characterizing the information-theoretic capacity scaling of wireless ad hoc networks using an exact channel model from Maxwell's equations in LOS environments. In dense networks, we establish the capacity scaling given as $\min\{n, \lambda^{-1}\}$ to within an arbitrarily small exponent. Hence, the capacity scaling is linear in $n$ if $\lambda=O(n^{-1})$. Otherwise, the capacity scaling is given as the DoF limit $\lambda^{-1}$ characterized by Franceschetti et al. In extended networks, the capacity scaling is given as $\min\{n, \frac{\sqrt{n}}{\lambda}\}$ to within an arbitrarily small exponent. Hence, the capacity scaling is linear in $n$ if $\lambda=O(n^{-1/2})$ and is given as the DoF limit $\frac{\sqrt{n}}{\lambda}$ characterized by Franceschetti et al. otherwise.

Since the converse is straightforward from the previous works in \cite{Ozgur:07,Franceschetti:09}, our main contribution is to show the achievability.  We note that under the far-field assumption, i.e., $\lambda$ is much smaller than the inter-node separation $\sqrt{\frac{A}{n}}$, where $\sqrt{A}$ denotes the network diameter, the DoF limit $\frac{\sqrt{A}}{\lambda}$ is in general higher than the throughput scaling $\sqrt{n}$ of the multihop via percolation theory of \cite{Franceschetti:07}. For achievability, we modify the HC scheme in \cite{Ozgur:07} according to an achievable MIMO rate between two node clusters. We show that the capacity of the MIMO channel between two node clusters is at least proportional to the minimum of the number of nodes in the cluster and the product of the ratio of the cluster diameter and $\lambda$ and the angular spread between clusters. In our modified HC scheme, only a subset of nodes in a cluster performs the MIMO transmission such that the number of participating nodes is proportional to the achievable MIMO rate, whereas all nodes in the cluster participate in the MIMO transmission in the HC scheme of~\cite{Ozgur:07}.

The organization of this paper is as follows. In Section \ref{sec:model}, the system model is presented. In Section \ref{sec:main}, we present the main theorems on the capacity scaling and their implications. In Section \ref{sec:modifiedHC}, a modified HC scheme is constructed according to an achievable MIMO rate between node clusters. We conclude this paper in Section \ref{sec:conclusion}.

The following notations will be used in the paper. $\mathcal{CN}(0,K)$ denotes the circularly symmetric complex Gaussian random vector with zero mean and covariance matrix of $K$. $\mathbb{R}$ and $\mathbb{N}$ denote the set of real numbers and the set of natural numbers, respectively.
$\E[\cdot]$ and $(\cdot)^*$ denote the expectation and conjugate transpose, respectively. $(\cdot)_m$ denotes the modulo-$m$ operation.
$(x)^+$ denotes the positive part of $x$, i.e.,
\begin{align*}
(x)^+=\begin{cases} x &\mbox{ if } x\geq 0\\
0 &\mbox{ if } x<0
\end{cases}.
\end{align*}
For two integers $u$ and $v$ such that $u\leq v$, $[u:v]$ denotes the set $\{u,u+1, \ldots, v\}$. For a set $\mathcal{S}$, $|\mathcal{S}|$ denotes the cardinality of the set. The logarithm function $\log$ is base 2 unless otherwise specified.

\section{System model} \label{sec:model}
There are $n$ uniformly and independently distributed nodes in a square of unit area (called a dense network) or a square of area $n$ (called an extended network). It is assumed that the node locations are fixed for the duration of the communication. Each node has an average transmit power constraint of $P$ and the network is allocated a total bandwidth $B$ around the carrier frequency $f_c\gg B$. The wavelength $\lambda=\frac{c}{f_c}$ is assumed to be much smaller than the average separation distance between neighbor nodes given as $\Theta(n^{-1/2})$ and $\Theta(1)$ for dense and extended networks, respectively. Furthermore, we assume a very mild lower bound on $\lambda$ such that $\lambda\geq n^{-\mu}$ for an arbitrarily large constant $\mu > 1/2$. We assume that $\lambda$ is a monotonically non-increasing function of $n$. This corresponds to using higher carrier frequencies to handle more traffic due to the increased number of nodes.
Every node is a source and a destination simultaneously, and the $n$ source--destination pairs are determined randomly. Every source wants to communicate to its destination at the same rate of $R(n,\lambda)$. The aggregate throughput $T(n,\lambda)$ of the network is given as $nR(n,\lambda)$.

We consider the LOS environment, i.e., no multi-path fading.\footnote{Our analysis can be extended to cases where there is multi-path fading. However, we believe that having a finite number of paths would not affect the throughput scaling laws.} From Maxwell's equations in far-fields, the discrete-time baseband-equivalent channel gain between nodes $k$ and $i$ at time $m$ is given as
\begin{align}
H_{ik}[m]&=\frac{\sqrt{G}}{d_{ik}[m]}\exp\left(-j\frac{2\pi}{\lambda}d_{ik}[m]\right)
\label{eqn:H_su}
\end{align}
where $j=\sqrt{-1}$, $d_{ik}[m]$ is the distance between nodes $k$ and $i$ at time $m$, and $G=\frac{\lambda^2 G_l}{16\pi^2}$  by Friis' formula, where $G_l$ is the product of the transmit and receive antenna gains.\footnote{A channel model with a path-loss exponent larger than two is considered in Appendix \ref{appendix:path_loss_two}.} Note that if $G_l$ is fixed, $G$ vanishes as $\lambda$ tends to zero. In extended networks, however, we assume that $G$ is a constant since we can increase $G_l$ proportional to $\lambda^{-2}$ without increasing the physical size of the antennas beyond a small fraction of the inter-node separation.\footnote{For each node, we can deploy $\Theta(\lambda^{-1})$ antennas vertically that form an antenna array of length $\Theta(1)$, which gives a vertical beamforming gain of $\Theta(\lambda^{-1})$. Hence, the product of the transmit and receive beamforming gains can be $\Theta(\lambda^{-2})$.} In dense networks, it is proper to assume that the node size is upper-bounded by $kn^{-1/2}$ for some constant $k$ since the network area is now fixed. Hence,  $G$ is assumed to be $\Theta(n^{-1})$ for dense networks because we can make $G_l$ proportional to $\lambda^{-2}n^{-1}$.\footnote{In dense networks, we can  vertically  deploy $\Theta(\lambda^{-1}n^{-1/2})$ antennas for each node that form an antenna array of length $\Theta(n^{-1/2})$.}

The discrete-time baseband-equivalent output $Y_i[m]$ at node $i$ at time $m$ is given as
\begin{equation*}
Y_i[m]=\sum^{n}_{k=1} H_{ik}[m]X_k[m] +Z_i[m]
\end{equation*}
where $X_k[m]$ is the discrete-time baseband-equivalent input at node $k$ at time $m$ and $Z_i[m]$ is the additive Gaussian noise $\mathcal{CN}(0,1)$ at node $i$ at time $m$. The channel state information (CSI) is available only at the receivers. From now on, we will omit the time index for notational convenience.

\section{Main result} \label{sec:main}
We first present a lower and an upper bound on the capacity scaling for dense networks in Theorems \ref{thm:dense} and \ref{thm:dense_ub}, respectively. In Theorems \ref{thm:extended} and \ref{thm:extended_ub}, we present a lower and an upper bound on the capacity scaling for extended networks, respectively.\footnote{We note that a similar result was also independently shown in~\cite{OLT:10} based on the same channel model as in~\cite{Lee:08} at the same time this paper was submitted. In this paper, we derive a lower bound on the MIMO transmission between two node clusters without any artificial assumptions, which is the key ingredient in the achievability,  whereas the work in \cite{OLT:10} assumed that interfering signals from other transmitting nodes in the network to the MIMO transmission are independent. In addition, the effect of $\lambda$ on $G$ is considered in this paper, but not in \cite{OLT:10}. }
\begin{theorem} \label{thm:dense}
Consider a network of $n$ nodes on a unit area, in which $n$ source--destination pairs are assigned arbitrarily. For any $\epsilon>0$, a scheme exists that achieves an aggregate throughput
\begin{align*}
T(n, \lambda)\geq K_{1\epsilon}\min\left\{\lambda^{-1},n\right\}^{1-\epsilon}
\end{align*}
with high probability,\footnote{With probability approaching 1 as $n$ tends to infinity.} where $K_{1\epsilon}$ is a positive constant independent of both $n$ and $\lambda$.
\end{theorem}
The aggregate throughput scaling in Theorem~\ref{thm:dense} can be achieved by the modified HC scheme constructed in Section~\ref{sec:modifiedHC}. Note that Theorem \ref{thm:dense} holds even if source--destination pairing is arbitrary.

In the following theorem, we show an upper bound on the throughput scaling. If the source--destination pairs can be determined according to the node locations, then an aggregate throughput scaling of $\Theta(n)$ would be achievable for any $\lambda$ by letting each of the source--destination pairs be nearest neighbors. Therefore, for the upper bound on the capacity scaling, we limit our interest to random source--destination pairing.

\begin{theorem}\label{thm:dense_ub}
Consider a network of $n$ nodes on a unit area, in which $n$ source--destination pairs are assigned randomly. The aggregate throughput in the network is upper-bounded as
\begin{align}
T(n, \lambda)\leq K_{2}\min\left\{\lambda^{-1}(\log\lambda^{-2})^2,n\log n\right\} \label{eqn:main_dense_ub}
\end{align}
with high probability,  where $K_2$ is a positive constant independent of both $n$ and $\lambda$.
\end{theorem}
The first term in the minimum in (\ref{eqn:main_dense_ub}) is the DoF limit shown in \cite{Franceschetti:09}.\footnote{ In \cite{Franceschetti:09}, the DoF limit in extended networks is studied and it is shown to be determined by the ratio of the network diameter and the wavelength. In  dense networks, the DoF limit can be obtained by rescaling the network size.} The second term in the minimum in (\ref{eqn:main_dense_ub}) is obtained from the fact that the transmission rate from a source to its destination is upper-bounded by the capacity of the single-input multiple-output (SIMO) channel between the source and the remaining nodes in the network (see, e.g., Theorem 3.1 in~\cite{Ozgur:07}).

Theorems \ref{thm:dense} and \ref{thm:dense_ub} establish the capacity scaling in dense networks to within an arbitrarily small exponent. To see the effect of $\lambda$ on the capacity scaling, let $\lambda=n^{-\beta}$ for $\beta\geq \frac{1}{2}$. Note that the condition $\beta\geq \frac{1}{2}$ is needed for the far-field approximation to hold. If $\beta\geq 1$, the capacity scaling is arbitrarily close to linear. If $\frac{1}{2}\leq\beta<1$, the capacity scaling is given as the DoF limit.

Now, we give an achievable aggregate throughput scaling in extended networks.
\begin{theorem} \label{thm:extended}
Consider a network of $n$ nodes on an area $n$, in which $n$ source--destination pairs are assigned arbitrarily. For any $\epsilon>0$, a scheme exists that achieves an aggregate throughput
\begin{align*}
T(n, \lambda)\geq K_{3\epsilon}\min\left\{\sqrt{n}\lambda^{-1},n\right\}^{1-\epsilon}
\end{align*}
with high probability, where $K_{3\epsilon}$ is a positive constant independent of both $n$ and $\lambda$.
\end{theorem}
The aggregate throughput scaling in Theorem~\ref{thm:extended} can be achieved by the modified HC scheme in Section~\ref{sec:modifiedHC}.

For random source--destination pairing, the following theorem shows an upper bound on the capacity scaling  whose exponent is arbitrarily close to that of the lower bound in Theorem~\ref{thm:extended}.
\begin{theorem}\label{thm:extended_ub}
Consider a network of $n$ nodes on an area $n$, in which $n$ source--destination pairs are assigned randomly. The aggregate throughput in the network is upper-bounded as
\begin{align}
T(n, \lambda)\leq
K_{4}\min\left\{\sqrt{n}\lambda^{-1}(\log (n\lambda^{-2}))^2,n\log n\right\}
\label{eqn:main_ext_ub}
\end{align}
with high probability,  where $K_4$ is a positive constant independent of both $n$ and $\lambda$.
\end{theorem}
The first term in the minimum in (\ref{eqn:main_ext_ub}) is the DoF limit shown in \cite{Franceschetti:09}, and the second term in the minimum in (\ref{eqn:main_ext_ub}) is obtained similarly as the derivation of the second term in the minimum in (\ref{eqn:main_dense_ub}).

Similarly as in dense networks, let $\lambda=n^{-\beta}$ for $\beta\geq 0$ to see how $\lambda$ affects the capacity scaling in extended networks. Note that $\beta=0$ means that $\lambda$ is a constant, regardless of $n$.
If $\beta\geq \frac{1}{2}$, the capacity scaling is arbitrarily close to \emph{linear}. If $0\leq \beta<\frac{1}{2}$, the capacity scaling is given as the DoF limit.

\begin{remark}
The exponent $\beta$ signifies the increase of $f_c$ to handle more traffic as $n$ increases. For example, consider a network with an area of $0.01\mbox{km}^2$ with $n=100$ and $f_c=300\mbox{MHz}$ ($\lambda=1\mbox{m}$). Then, the DoF limit is an order of 100, and hence, the network is not DoF limited. Now, assume that the network  size grows to an area of $1\mbox{km}^2$ with $n=10000$.  If $\beta=0$, i.e., the carrier frequency remains the same, then the network becomes DoF limited since the DoF limit is an order of 1000.  Now, if $\beta=1$, i.e., the carrier frequency is increased to $3\mbox{GHz}$ ($\lambda=0.1\mbox{m}$), then the network is not DoF limited since the DoF limit is now an order of 10000.
\end{remark}

\section{Modified Hierarchical Cooperation}\label{sec:modifiedHC}
In this section, Theorems \ref{thm:dense} and \ref{thm:extended} are proved by constructing a modified HC scheme.
Let us first consider a cooperative MIMO between two node clusters, which is the key to the construction of the modified HC scheme. Consider $N$ independently and uniformly distributed nodes in each of two horizontally aligned square areas with side length $D$ and distance $L\geq 2D$ between the centers, as shown in Fig. \ref{fig:MIMO}. Let $C_T$ and $C_R$ denote the left and right clusters of $N$ nodes in  Fig. \ref{fig:MIMO}, respectively.
The $N$-by-$N$ cooperative MIMO channel from $C_T$ to $C_R$ is given as
\begin{align}
Y=HX+W+Z \label{eqn:mimo_ch}
\end{align}
where $Y$ is the $N$-by-1 received vector at $C_R$, $H$ is the $N$-by-$N$ channel matrix from (\ref{eqn:H_su}), $X$ is the $N$-by-$1$ transmitted vector from $C_T$, $W$ is the $N$-by-1 external interference vector with covariance matrix $\Sigma$, and $Z$ is the $N$-by-1 additive Gaussian noise vector $\mathcal{CN}(0,I)$. Let $\rho_1\triangleq \frac{L^2}{NGP}\tr(\Sigma)$ and $\rho_2\triangleq\frac{L^4}{N(GP)^2}\tr(\Sigma^2)$. The following theorem presents an achievable MIMO rate from $C_T$ to $C_R$.\footnote{A more general version of Theorem~\ref{theorem:MIMO} was shown previously in Theorem 1 in~\cite{Lee:08}, where multiple antennas per node were assumed. $M$ in Theorem~\ref{theorem:MIMO} corresponds to the fourth term in the minimum of (4) in~\cite{Lee:08}, which was obtained based on an approximation and, therefore, differs slightly from $M$. In deriving Theorem~\ref{theorem:MIMO}, however, no approximation is used, and therefore, the result is now exact.} \footnote{For the simplicity of presentation, $C_T$ and $C_R$ are assumed to be horizontally aligned. However, the proof of Theorem \ref{theorem:MIMO} in Appendix \ref{appendix:MIMO} can be easily extended to cases where $C_T$ and $C_R$ are not horizontally aligned, which will result in the same conclusion as in Theorem  \ref{theorem:MIMO}.} The proof is in Appendix \ref{appendix:MIMO}.

\begin{figure}[t]
 \centering
{
\psfrag{t}[c]{$C_T$}
\psfrag{r}[c]{$C_R$}
\psfrag{l}[c]{$L$}
\psfrag{d}[c]{$D$}
  \includegraphics[width=85mm]{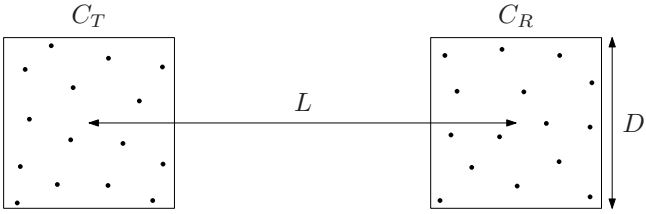}}
\caption{Cooperative MIMO between $C_T$ and $C_R$ with side length $D$ and distance $L\geq 2D$ between the centers } \label{fig:MIMO}
\end{figure}

\begin{theorem}
The capacity $C(H)$ of the cooperative MIMO channel from $C_T$ to $C_R$  is lower-bounded as
\begin{align*}
C(H)\geq N\frac{\delta^2(\rho_1+K_1'N)^2}{(\rho_2^{1/2}+(K_2'\max\{N^2,N^3M^{-1}\})^{1/2})^2}\cr
\times \log \left(1+\frac{\frac{GP}{L^2}((1-\delta)K_1'N-\delta\rho_1)}{1+\frac{GP}{L^2}\rho_1}\right)
\end{align*}
for any $0\leq \delta\leq 1$ with high probability as $N$ tends to infinity, where $K_1'$ and $K_2'$ are positive constants independent of $D$, $L$, $N$, and $\lambda$ and $M$ is given as
\begin{align}
M=\max\left\{1,\frac{D^2}{\lambda L} \left(1+\left(\log \frac{D^2}{\lambda L}\right)^+\right)^{-1}\right\}.\label{eqn:M}
\end{align}\label{theorem:MIMO}
\end{theorem}
We have the following corollaries for certain classes of $\Sigma$.
\begin{corollary}
If there is no external interference, i.e., $W=0$, the capacity $C(H)$ of the cooperative MIMO channel from $C_T$ to $C_R$  is lower-bounded as
\begin{align*}
C(H)\geq K_3'\min\{N,M\} \log\left(1+K_4'\frac{GP}{L^2}N\right)
\end{align*}
with high probability as $N$ tends to infinity, where $K_3'$ and $K_4'$ are positive constants independent of $D$, $L$, $N$, and $\lambda$.
\end{corollary}
\begin{proof}
We choose $\delta=\Theta(1)$ in Theorem \ref{theorem:MIMO}, e.g., $\delta=\frac{1}{2}$.
\end{proof}
\begin{corollary} \label{corollary:mimo}
If $\rho_1=O(sN)$ and $\rho_2=O(s^{\nu}\max\{N^2,N^3M^{-1}\})$, where $s=\Omega(1)$ and $\nu\geq 1$, the capacity $C(H)$ of the cooperative MIMO channel from $C_T$ to $C_R$  is lower-bounded as
\begin{align*}
C(H)\geq \frac{K_5'}{s^{\nu+4}}\min\{N,M\} \log\left(1+K_6'\frac{\frac{NGP}{L^2}(1-s^{-1})}{1+\frac{NGP}{L^2}s}\right)
\end{align*}
with high probability as $N$ tends to infinity, where $K_5'$ and $K_6'$ are positive constants independent of $D$, $L$, $N$, and $\lambda$.
\end{corollary}
\begin{proof}
We choose $\delta=\Theta(s^{-2})$ in Theorem \ref{theorem:MIMO}.
\end{proof}

Note that $M$ matches the DoF limit predicted in \cite{Poon:05,Franceschetti:09} given as the product of the normalized cluster diameter $\frac{D}{\lambda}$ and the angular spread $\frac{D}{L}$ between the clusters. %We note that the DoF limit can be the predicted from the previous results~\cite{Poon:05,Franceschetti:09}, however, it cannot be derived from those previous results and hence its proof is needed.

In the following subsections, Theorems \ref{thm:dense} and \ref{thm:extended} for dense and extended networks, respectively, are proved by constructing a modified HC scheme.
\subsection{Dense network}
Let us construct the modified HC scheme for dense networks consisting of $h$ hierarchy levels. For an area of $A'$, there are an order of $A'n$ nodes with high probability.\footnote{See Lemma 4.1 in~\cite{Ozgur:07} for the proof.} For simplicity, we assume that there are exactly $A'n$ nodes in our description of the scheme, but our results hold without such an assumption. Consider a $(h+1)$-tuple $(n_0,n_1,...,n_h)\in \mathbb{N}^{h+1}$ such that $n_h=n$ and $n_{k-1}\leq n_k$ for all $k\in [1:h]$ and a $h$-tuple $(m_0, m_1, ..., m_{h-1}) \in \mathbb{N}^{h}$ such that $m_{k}\leq n_{k}$ for all $k\in [0:h-1]$. For $k\in [0:h]$, let $A_k\triangleq \frac{n_k}{n}$ and $L_k=\sqrt{A_k}$. Consider a hierarchical structure of the network such that the network is divided into square areas of $A_{h-1}$, each of those square areas is again divided into smaller square areas of $A_{h-2}$, and so on, i.e., at the $k$-th hierarchy level for $k\in [1:h]$, each square area of $A_k$ is divided into smaller square areas of $A_{k-1}$.

Let $T_{k}(n_k,\lambda)$ for $k\in [0:h]$ denote the achievable throughput when a cluster of $n_k$ nodes operates as a network having its own $n_k$ source--destination pairs in an arbitrary manner. The following lemma gives $T_k(n_k,\lambda)$ as a function of $T_{k-1}(n_{k-1},\lambda)$ for $k\in [1:h]$.\footnote{Since $T_k(n_k,\lambda)$ for $k\in [1:h]$ has a recursive form, it also depends on $n_0,...,n_{k-1}, m_0,...,m_{k-1}$.}

\begin{lemma} \label{lemma:recursive}
Fix $k\in [1:h]$. Consider a cluster of $n_k$ nodes. If, for any two clusters $u$ and $v$ of $n_{k-1}$ nodes inside the cluster of $n_k$ nodes, a rate of $R_k$ is achievable with high probability for the MIMO communication from $m_{k-1}$ randomly chosen nodes in cluster $u$ to $m_{k-1}$ randomly chosen nodes in cluster $v$ when other nodes in the cluster of $n_k$ nodes are silent, we have
\begin{align}
T_{k}(n_k,\lambda) \geq
\frac{K''_1}{1+m_{k-1}/R_k}\frac{n_{k}m_{k-1}}{m_{k-1}n_{k-1}/T_{k-1}(n_{k-1},\lambda) + n_k} \label{eqn:Tk}
\end{align}
with high probability, where $K''_1$ is a positive constant independent of both $n$ and $\lambda$.
\end{lemma}

\begin{proof}
We construct a scheme for the cluster of $n_k$ nodes when it operates as a network having its own $n_k$ source--destination pairs in an arbitrary manner.  From now on, a cluster indicates a cluster of $n_{k-1}$ nodes inside the cluster of $n_k$ nodes unless otherwise specified. We randomly assign the indices $[1:n_{k-1}]$ to $n_{k-1}$ nodes in each cluster and let $\mathcal{A}_{u,v}$ for $u\in [1:\frac{n_k}{n_{k-1}}]$ and $v\in [1:n_{k-1}]$ denote the set $\{(v+i)_{n_{k-1}}+1|1\leq i\leq m_{k-1}\}$ of nodes in cluster $u$.

The scheme consists of three phases. Let us first explain the scheme briefly from the perspective of source $s$ in cluster $u$ and its destination $d$ in cluster $v$. In the first phase, source $s$ in cluster $u$ distributes its message to $\mathcal{A}_{u,s}$. In the second phase, $\mathcal{A}_{u,s}$ performs MIMO transmission to $\mathcal{A}_{v,d}$. In the last phase, destination $d$ in cluster $v$ collects quantized MIMO observations from $\mathcal{A}_{v,d}$ and decodes the message. The detailed operation in each phase is as follows.

\begin{figure}[t]
 \centering
{ \small
\psfrag{1}[c]{$1$}
\psfrag{2}[c]{$2$}
\psfrag{3}[c]{$3$}
\psfrag{4}[c]{$4$}
\psfrag{5}[c]{$5$}
\psfrag{6}[c]{$6$}
\psfrag{7}[c]{$7$}
\psfrag{8}[c]{$8$}
\psfrag{9}[c]{$9$}
  \includegraphics[width=40mm]{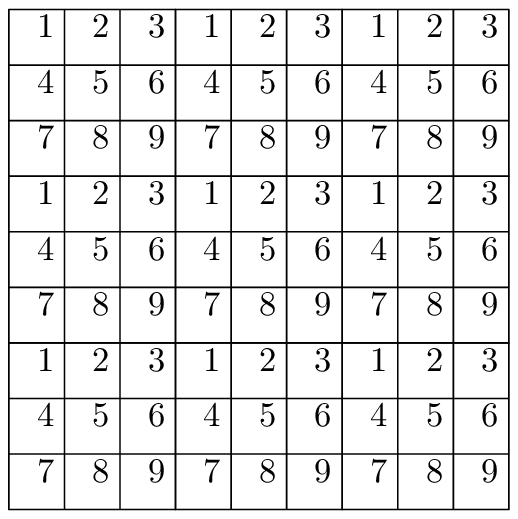}}
\caption{The big square and small squares represent a cluster of $n_k$ nodes and the clusters of $n_{k-1}$ nodes inside it, respectively. In Phases 1 and 3, the clusters of $n_{k-1}$ nodes operate in parallel according to the following 9-TDMA scheme: the total time of the phase is divided into 9 TDMA slots, and, in the $i$-th TDMA slot for $i\in [1:9]$, clusters marked with $i$ operate simultaneously while the other clusters are silent. } \label{fig:tdma}
\end{figure}

\begin{itemize}
\item Phase 1: Each cluster operates in parallel according to the 9-TDMA scheme of \cite{Ozgur:07} illustrated in Fig. \ref{fig:tdma}. Source $s$ in cluster $u$ distributes its message to $\mathcal{A}_{u,s}$, i.e., the message of $s$ is split into $m_{k-1}$ subblocks and each node in $\mathcal{A}_{u,s}$ receives one subblock. For a cluster, this can be done by setting up $m_{k-1}$ subphases, where $n_{k-1}$ source--destination pairs in each of the subphases are assigned as follows: in subphase $i\in [1:m_{k-1}]$, $\{(s, (s+i)_{n_{k-1}}+1)|s\in [1:n_{k-1}]\}$ is the set of $n_{k-1}$ source--destination pairs. Because $T_{k-1}(n_{k-1},\lambda)$ is achievable for a network of $n_{k-1}$ nodes, $n_{k-1} /T_{k-1}(n_{k-1},\lambda)$ time slots are needed for each subphase. Since there are $m_{k-1}$ subphases in each TDMA slot, Phase 1 needs a total of $9m_{k-1}n_{k-1}/T_{k-1}(n_{k-1},\lambda)$ time slots.

\item Phase 2: We perform successive MIMO transmissions for all source--destination pairs, i.e., MIMO transmission from $\mathcal{A}_{u,s}$ to $\mathcal{A}_{v,d}$ for source $s$ in cluster $u$ and destination $d$ in cluster $v$. Since a rate of $R_k$ is assumed to be achievable for each MIMO transmission, $m_{k-1}/R_k$ time slots are needed for each source--destination pair. Since we have $n_k$ source--destination pairs, a total of $n_km_{k-1}/R_k$ time slots are needed for Phase 2. After Phase 2, each node quantizes the MIMO observations at a fixed rate $Q$ subblocks per time slot.\footnote{From Appendix II in~\cite{Ozgur:07},  a strategy exists for each node to encode the observation of a MIMO transmission at a fixed rate $Q$ such that the resultant $m_{k-1}$-by-$m_{k-1}$ quantized MIMO channel has the same multiplexing gain as the original MIMO channel. }

\item Phase 3: Each cluster operates in parallel according to the 9-TDMA scheme of \cite{Ozgur:07} depicted in Fig. \ref{fig:tdma}. Destination $d$ in cluster $v$ collects the quantized observations of the MIMO transmission intended for it from $\mathcal{A}_{v,d}$ and then decodes the message. Note that each quantized MIMO observation consists of $Qm_{k-1}/R_k$ subblocks. By setting up $m_{k-1}$ subphases  for a cluster similarly as in Phase 1, where $n_{k-1}$ source--destination pairs are assigned in each of the subphases, a total of $(9Qm_{k-1}^2n_{k-1})/(R_kT_{k-1}(n_{k-1},\lambda))$ time slots are needed for Phase 3.
\end{itemize}
In total,
\begin{align*}
9m_{k-1}n_{k-1}/T_{k-1}(n_{k-1},\lambda)+n_km_{k-1}/R_k\\
+(9Qm_{k-1}^2n_{k-1})/(R_kT_{k-1}(n_{k-1},\lambda))
\end{align*}
time slots are needed to transport $n_k$ messages, i.e., $n_km_{k-1}$ subblocks. Hence, the constructed scheme yields an aggregate throughput of (\ref{eqn:recursive}), which proves Lemma \ref{lemma:recursive}.
\begin{figure*}\hrule
\begin{align}
T_{k}(n_k,\lambda)&=\frac{n_km_{k-1}}{9m_{k-1}n_{k-1}/T_{k-1}(n_{k-1},\lambda)+n_km_{k-1}/R_k+(9Qm_{k-1}^2n_{k-1})/(R_kT_{k-1}(n_{k-1},\lambda))}\cr
&\geq \frac{1}{9(Q+1)(1+m_{k-1}/R_k)}\frac{n_{k}m_{k-1}}{m_{k-1}n_{k-1}/T_{k-1}(n_{k-1},\lambda) + n_k} \label{eqn:recursive}
\end{align}\hrule
\end{figure*}

In the above explanation of the scheme, we focused on the modified operation from the scheme of~\cite{Ozgur:07} and the resulting scaling law of the throughput. The readers should refer to \cite{Ozgur:07} for a more detailed description of the scheme. However, taking those details into account does not change the throughput scaling. \end{proof}

The modified HC scheme is constructed recursively using the scheme in the proof of Lemma \ref{lemma:recursive} for the original network of $n$ nodes and using the multihop via percolation theory~\cite{Franceschetti:07} for clusters of $n_0$ nodes at the bottom hierarchy. Now, let us show an achievable throughput scaling using the modified HC scheme with $h$ hierarchy levels. Note that throughput achieved by the modified HC scheme depends on the choice of $(n_0,n_1,...,n_h)$ and $(m_0,m_1,...,m_{h-1})$.  First, we choose $m_{k-1}$ as $G_k$ for $k\in [1:h]$, where
\begin{align*}
G_k\triangleq \min\left\{n_{k-1}, \frac{n_{k-1}}{ (n_kn)^{\frac{1}{2}}\lambda\log\lambda^{-1} } \right\}.
\end{align*}
For the modified HC scheme with the above choice of $(m_0,m_1,...,m_{h-1})$, the following lemma shows that a rate of $R_k=\Theta(G_k/(\log n)^7)$ is achievable for the MIMO transmissions in Phase 2 at the $k$-th hierarchy level  for $k\in [1:h]$.

\begin{lemma} \label{lemma:Gk}
In Phase 2 at the $k$-th hierarchy level of the modified HC scheme for $k\in [1:h]$, a rate of $R_k=\Theta(G_k/(\log n)^7)$ is achievable for the MIMO transmissions between clusters of $G_k$ nodes.
\end{lemma}

\begin{proof}
Fix $k\in[1:h]$. In Phase 2 at the $k$-th hierarchy level, we let each transmitting cluster of $G_k$ nodes use a randomly generated Gaussian code according to $\mathcal{CN}(0,P'I)$, where
\begin{align}
P'=\frac{L_k^2}{GG_k}P=\frac{n_k}{G_k}P.\label{eqn:P}
\end{align}
This satisfies the average power constraint of $P$ per node because each node participates in the MIMO transmission for $\frac{G_k}{n_k}$ fraction of time in Phase 2.

Consider the MIMO transmission from cluster $C_T$ of $G_k$ nodes to cluster $C_R$ of $G_k$ nodes inside cluster $V$ of $n_k$ nodes in Phase 2 at the $k$-th hierarchy level of the modified HC scheme. To prove that the capacity of the MIMO channel from $C_T$ to $C_R$ is at least linear in $G_k/(\log n)^7$, we use Corollary \ref{corollary:mimo}. By adopting the notations for Corollary \ref{corollary:mimo}, let $D$ and $L$ denote the side length of $C_T$ and $C_R$ and the distance between the centers, respectively, and  let $M$ be given as (\ref{eqn:M}). The MIMO transmission from $C_T$ to $C_R$ is interfered by the MIMO transmission by $G_k$ nodes in each cluster of $n_k$ nodes that operates simultaneously with $V$. Let $\mathcal{U}_V$ denote the set of clusters of $n_k$ nodes that operate simultaneously with $V$. Then, $\mathcal{U}_V$ can be split into subgroups according to their distance to $V$ such that the $i$-th subgroup $\mathcal{U}_V(i)$ contains $8i$ or less clusters of $n_k$ nodes and the distance between the centers of $V$ and each cluster in $\mathcal{U}_V(i)$ is greater than or equal to $(3i)L_k$ for $i=1,2, \ldots$, as illustrated in Fig. \ref{fig:interference}. The number of such subgroups can be simply bounded by $n/n_k$. Let $|\mathcal{U}_V(i)|$ denote the number of clusters of $n_k$ nodes in $\mathcal{U}_V(i)$.
Then, the MIMO channel from $C_T$ to $C_R$ is given as (\ref{eqn:mimo_ch}), in which $G_k$ is substituted for $N$ and the interference $W$ is given as
$W=\sum_{i=1}^{n/n_k}\hat{H}_i\hat{X}_i$, where $\hat{H}_i$ is the $G_k$-by-$(|\mathcal{U}_V(i)|G_k)$ channel matrix from $\mathcal{U}_V(i)$ to $C_R$ and $\hat{X}_i$ is the $(|\mathcal{U}_V(i)|G_k)$-by-1 transmitted vector from $\mathcal{U}_V(i)$.
\begin{figure}[t]
 \centering
{ \small
\psfrag{V}[c]{$V$}
  \includegraphics[width=40mm]{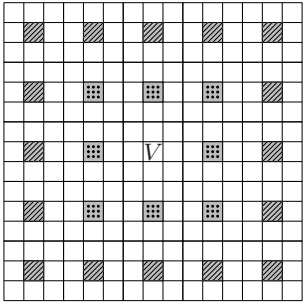}}
\caption{For cluster $V$, clusters that operate simultaneously with $V$ according to the 9-TDMA scheme are represented as shaded. The set of shaded clusters with dots represents $\mathcal{U}_V(1)$, and the set of shaded clusters with slash lines represents $\mathcal{U}_V(2)$. } \label{fig:interference}
\end{figure}

Now, let us show that the covariance matrix $\Sigma$ of $W$ satisfies the conditions in Corollary \ref{corollary:mimo} for $s=\log n$ and $\nu=2$. Let $\hat{F}_i=\frac{(3i)L_k}{\sqrt{G}}\hat{H}_i$.
Then, we have
\begin{align*}
\rho_1&=\frac{L^2}{G_kGP'}\tr(\Sigma)\cr
&=\frac{1}{G_k}\sum_{i=1}^{n/n_k}\tr\left(\frac{L^2}{G}\hat{H}_i\hat{H}_i^*\right)\cr
&\leq\frac{1}{G_k}\sum_{i=1}^{n/n_k}\frac{1}{(3i)^2}\tr(\hat{F}_i\hat{F}_i^*)
\end{align*}
and
\begin{align*}
\rho_2&=\frac{L^4}{G_k(GP')^2}\tr(\Sigma^2)\\
&=\frac{L^4}{G_kG^2}\tr\left(\left(\sum_{i=1}^{n/n_k}\hat{H}_i\hat{H}_i^*\right)^2\right)\\
&\overset{(a)}{\leq} \frac{L^4}{G_kG^2}\left(\sum_{i=1}^{n/n_k} \tr^{1/2}(\hat{H}_i\hat{H}_i^*\hat{H}_i\hat{H}_i^*)\right)^2\\
&=\frac{1}{G_k}\left(\frac{L}{L_k}\right)^4\left(\sum_{i=1}^{n/n_k} \frac{1}{(3i)^2}\tr^{1/2}(\hat{F}_i\hat{F}_i^*\hat{F}_i\hat{F}_i^*)\right)^2
\end{align*}
where $(a)$ is from the following lemma, which is a direct consequence of the matrix trace inequality in \cite{YangYangTeo:78}.
\begin{lemma}\label{lemma:tr_inequality}
If $A_i$'s are positive semidefinite matrices, then $\tr(\sum_i A_i)^2\leq (\sum_i \tr^{1/2}(A_i^2))^2$.
\end{lemma}

By applying similar bounding techniques as those for $\tr(FF^*)$ and $\tr(FF^*FF^*)$ in Appendix \ref{appendix:MIMO} to $\tr(\hat{F}_i\hat{F}_i^*)$ and $\tr(\hat{F}_i\hat{F}_i^*\hat{F}_i\hat{F}_i^*)$, we can show \[\tr(\hat{F}_i\hat{F}_i^*) =O(|\mathcal{U}_V(i)| G_k^2)\] and \[\tr(\hat{F}_i\hat{F}_i^*\hat{F}_i\hat{F}_i^*)=O(|\mathcal{U}_V(i)|^2\max\{G_k^3, G_k^4\hat{M}^{-1}\})\] with high probability, where $\hat{M}$ is given as
\[\hat{M}=\max\left\{1,\frac{D^2}{\lambda L_k} \left(1+\left(\log \frac{D^2}{\lambda L_k}\right)^+\right)^{-1}\right\}.\]
Because $|\mathcal{U}_V(i)|\leq 8i$ and $\left(\frac{L}{L_k}\right)^4\hat{M}^{-1}\leq M^{-1}$, we have \[\frac{L^2}{G_kGP'}\tr(\Sigma)=O((\log n) G_k)\] and \[\frac{L^4}{G_k(GP')^2}\tr(\Sigma^2)=O((\log n)^2\max\{G_k^2, G_k^3M^{-1}\}).\] Hence, the conditions in Corollary \ref{corollary:mimo} are satisfied for $s=\log n$ and $\nu=2$. From Corollary \ref{corollary:mimo} for $s=\log n$ and $\nu=2$, the capacity $C(H)$ of the MIMO channel from $C_T$ to $C_R$ is lower-bounded as
\begin{align*}
C(H)&\geq K_5'\frac{\min\{G_k,M\}}{(\log n)^6}\log\left(1+K_6' \frac{\frac{G_kGP'}{L^2}(1-\frac{1}{\log n})}{1+\frac{G_kGP'}{L^2}\log n}\right)\\
&\overset{(a)}{\geq} K_7'\frac{\min\{G_k,M\}}{(\log n)^7}
\end{align*}
for some constant $K_7'$ with high probability, where $(a)$ is from the choice of $P'$ in (\ref{eqn:P}). Furthermore, because $\frac{D^2}{\lambda L}=O(\lambda^{-1})$ and $\frac{D^2}{\lambda L}=\Omega(\frac{n_{k-1}}{\lambda (n_kn)^{\frac{1}{2}}})$, we have $M=\Omega(\frac{n_{k-1}}{ (n_kn)^{\frac{1}{2}}\lambda\log\lambda^{-1} })$. Hence, we have $C(H)=\Omega(G_k/(\log n)^7)$, which proves Lemma \ref{lemma:Gk}.
\end{proof}

Now, by substituting $G_k$ and $K_2''(\log n)^7$ for $m_{k-1}$ and $m_{k-1}/R_k$ in (\ref{eqn:Tk}), where $K_2''$ is a positive constant independent of both $n$ and $\lambda$, we have the recursive form of $T_k(n_k,\lambda)$ for $k\in [1:h]$ for the modified HC scheme given as
\begin{align}
T_{k}(n_k,\lambda) \geq \frac{K''_3}{(\log n)^7}\frac{n_{k}G_k}{n_{k-1}G_k/T_{k-1}(n_{k-1},\lambda) + n_k} \label{eqn:Tk2}
\end{align}
where $K_3''$ is a positive constant independent of both $n$ and $\lambda$.

The following lemma gives an achievable throughput scaling using  the modified HC scheme with $h$ hierarchy levels when we choose $(n_0, n_1, ..., n_{h})$ that maximizes (\ref{eqn:Tk2}) for $k\in[1:h]$. The proof is at the end of the present section.
\begin{lemma}\label{lemma:throughput_h}
In dense networks, the modified HC scheme with $h$ hierarchy levels achieves
\begin{align*}
T_h(n,\lambda)\geq \frac{C_h}{(\log n)^{7h+1}}\frac{n^{\delta_{b(n,\lambda,h)}}}{(\lambda \log \lambda^{-1})^{\tau_{b(n,\lambda,h)}}}
\end{align*}
with high probability, where $C_h$ is a positive constant independent of both $n$ and $\lambda$,
\begin{align*}
&b(n,\lambda,h)\cr
&\triangleq\begin{cases}
h+1&\mbox{if } \log_n (\lambda\log\lambda^{-1})\leq -\Lambda(h)   \\
k  &\mbox{if } -\Lambda(k)< \log_n (\lambda\log\lambda^{-1})\leq -\Lambda(k-1) \\
&\mbox{~~~~~~~~~~~~~~~~~~~~~~~~~~for some }k\in [2:h] \\
1  &\mbox{if } -\Lambda(1)< \log_n (\lambda\log\lambda^{-1})
\end{cases}, \end{align*}
and
\begin{align*}
\delta_{u}\triangleq\frac{u2^{h-u}}{3^{1+h-u} + 2^{h-u}(u-1)},~\tau_{u}\triangleq\frac{3^{1+h-u} - 2^{1+h-u}}{3^{1+h-u} + 2^{h-u}(u-1)}
\end{align*}
for $u\in [1:h+1]$, where $\Lambda(v)\triangleq  \frac{3^{h-v}(3+v)-2^{h-v}}{3^{h-v}(4+v)-2^{1+h-v}}$ for $v\in [1:h]$.
\end{lemma}
The following corollary is obtained straightforwardly from Lemma \ref{lemma:throughput_h}.
\begin{corollary} \label{corollary:hc}
In dense networks, the modified HC scheme with $h$ hierarchy levels achieves
\begin{align*}
T_h(n,\lambda)\geq \frac{C_h'}{(\log n)^{7h+1}} \min\left\{\frac{n^{\delta_k}}{\left(\lambda\log\lambda^{-1}\right)^{\tau_k}}\bigg|1\leq k\leq h+1\right\}
\end{align*}
with high probability,  where $C_h'$ is a positive constant independent of both $n$ and $\lambda$.
\end{corollary}

Now we are ready to prove Theorem \ref{thm:dense}.

\subsubsection*{Proof of Theorem \ref{thm:dense}}
Fix $\epsilon>\epsilon'>0$. Let $h$  be the smallest integer such that $h>\frac{8}{\epsilon'}$ and let $n$ be the smallest integer such that $(7h+1)\log_n\log n<\frac{\epsilon'}{4}$ and $\log_{n^{1/2}}\log n^\mu<\frac{\epsilon-\epsilon'}{1-\epsilon'}$. Let us  define functions $y_0(x)$ and $y_{k}(x)$ for $k\in [1: h+1]$ for $x\leq -\frac{1}{2}$ as
\begin{align*}
y_0(x)&=(1-\epsilon')\min\{-x,1\},\\
y_{k}(x)&= \delta_{k}-\tau_{k}x-(7h+1)\log_n\log n.
\end{align*}

Fix $k\in [1:h+1]$. We will show that $y_{k}(x)$ is larger than $y_0(x)$ for all $x\leq -\frac{1}{2}$.  Let us first show that $y_{k}(-1)>1-\frac{\epsilon'}{2}$. $y_{k}(-1)$ is given as
\begin{align*}
y_{k}(-1)&=\delta_{k}+\tau_{k}-(7h+1)\log_n\log n \\
&> \delta_{k}+\tau_{k}-\frac{\epsilon'}{4} \\
&=1-\frac{1}{3\left(\frac{3}{2}\right)^{h-k}+k-1}-\frac{\epsilon'}{4}.
\end{align*}
If $1\leq k<\frac{h}{2}$, we have
\begin{align*}
3\left(\frac{3}{2}\right)^{h-k}  +k-1\geq 3\left(\frac{3}{2}\right)^{h-k} > 3\left(\frac{3}{2}\right)^{\frac{h}{2}}>\frac{h}{2}>\frac{4}{\epsilon'}.
\end{align*}
If $\frac{h}{2}\leq k\leq h+1$, we get
\begin{align*}
3\left(\frac{3}{2}\right)^{h-k}+k-1>k \geq \frac{h}{2}>\frac{4}{\epsilon'}.
\end{align*}
Thus, we conclude that $y_{k}(-1)>1-\frac{\epsilon'}{2}$. Now we are ready to show $y_{k}(x)> y_0(x)$ for all $x\leq -\frac{1}{2}$. Note that $0\leq \tau_{k}<1$. For $x< -1$,
\begin{align*}
y_{k}(x) = y_{k} (-1)-\tau_{k}(x+1) \geq  y_{k} (-1) > 1-\epsilon' = y_0(x).
\end{align*}
For $-1\leq x\leq -\frac{1}{2}$,
\begin{align*}
y_{k}(x)&=y_{k}(-1)-\tau_{k}(x+1)\cr
&\geq y_{k}(-1)-(x+1)\cr
&>-(1-\epsilon')x\cr
&=y_0(x).
\end{align*}
Hence, we prove that $\min\left\{y_{k}(x):1\leq k\leq h+1\right\}>y_0(x)$ for all $x\leq -\frac{1}{2}$.
By letting $x=\log_n(\lambda\log\lambda^{-1})$, we equivalently prove that the achievable rate of the modified HC scheme with $h$ hierarchy levels in Corollary \ref{corollary:hc} is lower-bounded as
\begin{align*}
T_h(n, \lambda)&\geq  \frac{C_h'}{(\log n)^{7h+1}} \min\left\{\frac{n^{\delta_k}}{\left(\lambda\log\lambda^{-1}\right)^{\tau_k}}\bigg|1\leq k\leq h+1\right\}\\
&>C_h'\min\left\{\frac{\lambda^{-1}}{\log\lambda^{-1}},n\right\}^{1-\epsilon'}.
\end{align*}
Now, we have
\begin{align*}
T_h(n, \lambda)&>C_h'\min\left\{\frac{\lambda^{-1}}{\log\lambda^{-1}},n\right\}^{1-\epsilon'}\\
&=C_h'\min\left\{(\lambda^{-1})^{1-\log_{\lambda^{-1}} \log\lambda^{-1}},n\right\}^{1-\epsilon'}\\
&\overset{(a)}{>}C_h'\min\left\{(\lambda^{-1})^{1-\log_{n^{1/2}} \log n^{\mu}},n\right\}^{1-\epsilon'}\\
&>C_h'\min\left\{\lambda^{-1},n\right\}^{1-\epsilon}
\end{align*}
where $(a)$ is because $n^{-\mu}<\lambda<n^{-1/2}$.
Hence, Theorem \ref{thm:dense} is proved. \endproof

\subsection{Extended network}
In extended networks, both $\sqrt{G}$ and the distance between nodes is increased by a factor of $\sqrt{n}$ as compared to those in dense networks. Hence, for the same transmit power, the received power at each node remains the same as in dense networks. By rescaling the space, let us consider an extended network as an equivalent dense network on a unit area but with the wavelength reduced to $\lambda n^{-1/2}$.
Since the wavelength is given as $\lambda n^{-1/2}$ in the equivalent dense network, Theorem \ref{thm:extended} is proved.

\subsubsection*{Proof of Lemma~\ref{lemma:throughput_h}}
First, consider the case of $b(n,\lambda,h)=h+1$. Since $\Lambda(h)=1$, this implies $\lambda\log\lambda^{-1}\leq n^{-1}$. In this case, $G_k$ is $n_{k-1}$, and hence, the recursive form of $T_k(n_k,\lambda)$  in (\ref{eqn:Tk2}) becomes
\begin{align}
T_{k}(n_k,\lambda) &\geq \frac{K_3''}{(\log n)^7}\frac{n_{k-1}n_k}{n_{k-1}^2/T_{k-1}(n_{k-1},\lambda) + n_k}\label{eqn:recursive_linear}
\end{align}
for all $k\in [1:h]$. Note that $T_0(n_0,\lambda)=\Theta(\frac{\sqrt{n_0}}{\log n})$ by using the multihop via percolation theory~\cite{Franceschetti:07} for the cooperation for the clusters of $n_0$ nodes.\footnote{In \cite{Franceschetti:07}, a path-loss exponent larger than two is considered and a multihop via percolation theory is shown to achieve $\Theta(\sqrt{n})$. For the path-loss exponent equal to two, however, it achieves $\Theta(\frac{\sqrt{n}}{\log n})$ due to the interference power proportional to $\log n$. } By choosing $n_{k-1}=n_k^{\frac{k+1}{k+2}}$ that maximizes (\ref{eqn:recursive_linear}) for $k\in [1:h]$, $T_h(n,\lambda)\geq \frac{C_h}{(\log n)^{7h+1}}n^{\frac{h+1}{h+2}}$ is obtained. Because $\delta_{h+1}=\frac{h+1}{h+2}$ and $\tau_{h+1}=0$, Lemma \ref{lemma:throughput_h} is proved for the case of $b(n,\lambda,h)=h+1$.

Next, consider the case of $b(n,\lambda,h)=h'$ for some $h'\in [1:h]$. Let us first assume that $G_k$ is $n_{k-1}$ for $k\in [1:h'-1]$ and is $\frac{n_{k-1}}{(n_kn)^{\frac{1}{2}}\lambda\log \lambda^{-1}}$ for $k\in [h':h]$. For the choice of $n_0,n_1,...,n_{h-1}$ that maximizes (\ref{eqn:Tk2}) under this assumption, we will show that the range of $\lambda$ where the assumption is valid is the same as the range of $\lambda$ corresponding to $b(n,\lambda,h)=h'$ in Lemma \ref{lemma:throughput_h}.

Since $G_k$ is assumed to be $n_{k-1}$ for $k\in [1: h'-1]$, we obtain
\begin{align}
T_{h'-1}(n_{h'-1},\lambda)\geq \frac{C_{h'-1}}{(\log n)^{7(h'-1)+1}}n_k^{\frac{h'}{h'+1}}. \label{eqn:th_h_1}
\end{align}
For $k\in [h':h]$, $G_k$ is assumed to be $\frac{n_{k-1}}{(n_kn)^{\frac{1}{2}}\lambda\log \lambda^{-1}}$, and hence, the recursive form of $T_k(n_k,\lambda)$  in (\ref{eqn:Tk2}) is given as
\begin{align}
&T_{k}(n_k,\lambda) \geq \frac{K_3''}{(\log n)^7}\cr
& ~~\times\frac{n_{k-1}n_k}{n_{k-1}^2/T_{k-1}(n_{k-1},\lambda) + n_k^{3/2} ((\lambda\log\lambda^{-1})^2n)^{\frac{1}{2}} }.\label{eqn:recursive_bottleneck}
\end{align}
Let us assume that $T_{k}(n_k,\lambda)$ for $k\in [h'-1:h]$ has the form of $\frac{C_{k}}{(\log n)^{7k+1}}\frac{n_{k}^{\alpha_{h',k}}}{\left(\left(\lambda\log\lambda^{-1}\right)^2n\right)^{\beta_{h',k}}}$
for some positive constants $C_{k}$, $\alpha_{h',k}$, and $\beta_{h',k}$ independent of both  $n$ and $\lambda$. Then, the recursive formulas $\alpha_{h',k}=\frac{\alpha_{h',k-1}+1}{2(2-\alpha_{h',k-1})}$ and
$\beta_{h',k}=\frac{1-\alpha_{h',k-1}+2\beta_{h',k-1}}{2(2-\alpha_{h',k-1})}$ are obtained by choosing $n_{k-1}$ as
\begin{align*}
n_{k-1}=n_k^{\frac{3}{2(2-\alpha_{h',k-1})}}((\lambda\log\lambda^{-1})^2n)^{\frac{1-2\beta_{h',k-1}}{2(2-\alpha_{h',k-1})}}
\end{align*}
that maximizes (\ref{eqn:recursive_bottleneck}) for $k\in[h':h]$. Using the conditions $\alpha_{h',h'-1}=\frac{h'}{h'+1}$ and $\beta_{h',h'-1}=0$ from (\ref{eqn:th_h_1}), $\alpha_{h',k}$ and $\beta_{h',k}$ for $k\in [h':h]$ are given as
\begin{align*}
\alpha_{h',k}&=\frac{3^{1+k-h'} + 2^{1+k-h'}(h'-1)}{3^{1+k-h'}2 + 2^{1+k-h'}(h'-1)},\\
\beta_{h',k}&=\frac{3^{1+k-h'} - 2^{1+k-h'}}{3^{1+k-h'}2 + 2^{1+k-h'}(h'-1)}.
\end{align*}
Because $n_h=n$, $n_k$ for $k\in [h'-1: h]$ is given as (\ref{eqn:n_k_explicit}).
\begin{figure*}\hrule
\begin{align}
n_k&=n^{\prod_{j=k+1}^h\left(\frac{3}{2(2-\alpha_{h',j-1})}\right)}\left(\left(\lambda \log\lambda^{-1}\right)^2n\right)^{\sum_{j={k+1}}^h\left(\frac{1-2\beta_{h',j-1}}{2(2-\alpha_{h',j-1})}\right)\prod_{i=k+1}^{j-1}\left(\frac{3}{2(2-\alpha_{h',i-1})}\right)}\cr
&=n^{\frac{3^{1+h-h'}+h'2^{1+k-h'} 3^{h-k}-2^{h-h'}(1+h')}{3^{1+h-h'}+2^{h-h'}(-1+h')}}\left(\lambda\log\lambda^{-1}\right)^{\frac{\left(2^{1+k-h'} 3^{h-k}-2^{1+h-h'}\right)(1+h')}{3^{1+h-h'}+2^{h-h'}(-1+h')}} \label{eqn:n_k_explicit}
\end{align}\hrule
\end{figure*}

Now, the range of $\lambda$ that makes the assumption, i.e., $G_k$ is $n_{k-1}$ for $k\in [1:h'-1]$ and is $\frac{n_{k-1}}{(n_kn)^{\frac{1}{2}}\lambda\log \lambda^{-1}}$ for $k\in [h':h]$, valid is given as
\begin{align}
\begin{cases}
(nn_{h'})^{-1/2}< \lambda \log \lambda^{-1} \leq (nn_{h'-1})^{-1/2} & \mbox{ if $h'\in [2:h]$,} \\
(nn_{1})^{-1/2}< \lambda \log \lambda^{-1} & \mbox{ if $h'= 1$.}
\end{cases}\label{eqn:bottleneck}
\end{align}
By using $n_{h'}$ and $n_{h'-1}$ from (\ref{eqn:n_k_explicit}), we can show that the range of $\lambda$ in (\ref{eqn:bottleneck}) is the same as the range of $\lambda$ corresponding to $b(n,\lambda,h)=h'$ in Lemma \ref{lemma:throughput_h}.
Hence, we prove that for $b(n,\lambda,h)=h'$, the modified HC scheme with $h$ levels achieves
\begin{align*}
T_h(n,\lambda)&\geq \frac{C_h}{(\log n)^{7h+1}}\frac{n^{\alpha_{h',h}}}{((\lambda \log \lambda^{-1})^2n)^{\beta_{h',h}}}\cr
&=\frac{C_h}{(\log n)^{7h+1}}\frac{n^{\alpha_{h',h}-\beta_{h',h}}}{(\lambda \log \lambda^{-1})^{2\beta_{h',h}}}.
\end{align*}
Since $\delta_{h'}=\alpha_{h',h}-\beta_{h',h}$ and $\tau_{h'}=2\beta_{h',h}$, Lemma \ref{lemma:throughput_h} is proved for the case of $b(n,\lambda,h) \in [1:h]$.
\endproof

\section{Conclusion}\label{sec:conclusion}
We characterized the information-theoretic capacity scaling of wireless ad hoc networks from Maxwell's equations without any artificial assumptions. The capacity scaling is given as the minimum of the number of nodes and the DoF limit given as the ratio of the network diameter and the wavelength. Accordingly, a network becomes DoF-limited if $\lambda=\Omega(n^{-1})$ in dense networks and $\lambda=\Omega(n^{-1/2})$ in extended networks. Our results indicate that the linear throughput scaling in \cite{Ozgur:07} that was shown under the i.i.d. channel phase assumption is indeed achievable to within an arbitrarily small exponent in the non DoF-limited regime. In the DoF-limited regime, the DoF limit characterized by Franceschetti et al. in \cite{Franceschetti:09} that generally has higher scaling than that of the multihop scheme can be achieved to within an arbitrarily small exponent by using the modified HC scheme.

We also considered a channel model with a path-loss exponent $\alpha$ larger than two. In dense networks, the throughput scaling using the modified HC scheme for $\alpha>2$ remains the same as when $\alpha=2$. However, the throughput scaling using the modified HC scheme is decreased for $\alpha>2$ in extended networks due to the power limitation. This suggests, as a further work, an upper bound considering both the DoF limitation due to the channel correlation and the power limitation due to the power attenuation over the distance.

\appendices
\section{Proof of Theorem \ref{theorem:MIMO}}
\label{appendix:MIMO}
The capacity $C(H)$ of the MIMO channel from $C_T$ to $C_R$ is lower-bounded as
\begin{align}
C(H)&=\max_{f(x):\E[|X_i|^2]\leq P}I(X;Y)\cr
&\overset{(a)}{\geq} I(X_{G};Y) \cr
&\overset{(b)}{\geq} \log \frac{\det(I+\Sigma+PHH^*)}{\det(I+\Sigma)}\cr
&=\log \frac{\prod_{i=1}^N(1+\frac{GP}{L^2}\kappa_i)}{\prod_{i=1}^N(1+\frac{GP}{L^2} \chi_i)}\cr
&\overset{(c)}{\geq}\log \prod_{i=1}^N\frac{(1+\frac{GP}{L^2}\kappa_i)}{1+\frac{GP}{L^2} \E[\chi]}\cr
&=\sum_{i=1}^N\log \frac{1+\frac{GP}{L^2}\kappa_i}{1+\frac{GP}{L^2} \E[\chi]}\cr
&= N \E\left[\log \frac{1+\frac{GP}{L^2}\kappa}{1+\frac{GP}{L^2} \E[\chi]}\right]\cr
&\geq N\P\left(\kappa>(1-\delta)\E[\kappa]\right)\log \frac{1+\frac{GP}{L^2}(1-\delta)\E[\kappa]}{1+\frac{GP}{L^2} \E[\chi]}\cr
&\overset{(d)}{=} N\P\left(\kappa>(1-\delta)\E[\kappa]\right)\cr
&~\times \log \left(1+ \frac{\frac{GP}{L^2}((1-\delta)\E[\gamma]-\delta\E[\chi])}{1+\frac{GP}{L^2} \E[\chi]}\right)\cr
&\overset{(e)}{\geq} N\frac{\delta^2\E[\kappa]^2}{\E[\kappa^2]}\log\left(1+ \frac{\frac{GP}{L^2}((1-\delta)\E[\gamma]-\delta\E[\chi])}{1+\frac{GP}{L^2} \E[\chi]}\right)\cr
&\overset{(f)}{\geq}N\frac{\delta^2(\E[\chi]+\E[\gamma])^2}{(\E^{1/2}[\chi^2]+\E^{1/2}[\gamma^2])^2}\cr
&~\times\log\left(1+ \frac{\frac{GP}{L^2}((1-\delta)\E[\gamma]-\delta\E[\chi])}{1+\frac{GP}{L^2} \E[\chi]}\right)\label{eqn:slow5}
\end{align}
for any $0\leq \delta \leq 1$, where $X_{G}$ is $\mathcal{CN}(0, PI)$, $\kappa$ is chosen uniformly among the eigenvalues $\kappa_i, i=1,\ldots, N$ of $\frac{L^2}{GP}(\Sigma+PHH^*)$, $\chi$ is chosen uniformly among the eigenvalues $\chi_i, i=1,\ldots, N$ of $\frac{L^2}{GP}\Sigma$, and $\gamma$ is chosen uniformly among the eigenvalues $\gamma_i, i=1,\ldots, N$ of $\frac{L^2}{G}HH^*$. $(a)$ is from choosing the input $X$ as $X_{G}$, $(b)$ is because assuming Gaussian interference minimizes the mutual information for given noise and interference covariance matrices \cite{Ihara:78, Cover:01}, $(c)$ is because the geometric mean is upper-bounded by the arithmetic mean, $(d)$ is because $\E[\kappa]=\E[\chi]+\E[\gamma]$, $(e)$ is from the Paley-Zygmund inequality~\cite{Ozgur:07,Kahane:85}, and $(f)$ is from Lemma \ref{lemma:tr_inequality}.

Note that $\E[\chi]=\rho_1$ and $\E[\chi^2]=\rho_2$. To get a lower bound on (\ref{eqn:slow5}), we need a lower bound on $\E[\gamma]$ and an upper bound on $\E[\gamma^2]$.  Let $F\triangleq \frac{L}{\sqrt{G}}H$. Then, $F_{ik}=a_{ik}\exp(-j2\pi\frac{d_{ik}}{\lambda})$, where $a_{ik}=\frac{L}{d_{ik}}$. Note that constants $a_{\min}$ and $a_{\max}$ exist independent of $D$ and $L$ such that $a_{\min}\leq a_{ik}\leq a_{\max}$ for all $i,k\in [1:N]$. First, $\E[\gamma]$ is given as
\begin{align*}
\E[\gamma]&=\frac{1}{N}\tr\left(FF^{*}\right)\cr
&=\frac{1}{N}\sum^{N}_{i,k=1}|F_{ik}|^2.
\end{align*}
Since $a_{\min}^2\leq |F_{ik}|^2\leq a_{\max}^2$, we have $\E[\gamma]=\Theta(N)$.

Next, $\E[\gamma^2]$ is upper-bounded as
\begin{align*}
\E[\gamma^2]&=\frac{1}{N}\tr\left(FF^{*}FF^{*}\right)\cr
&=\frac{1}{N}\sum^{N}_{i,j,k,l=1}F_{ik} F_{il}^*F_{jl}F_{jk}^*\cr
&=\frac{1}{N}\sum_{(i,j,k,l)\in \Psi_1}F_{ik} F_{il}^*F_{jl}F_{jk}^*\cr
&~+\frac{1}{N}\sum_{(i,j,k,l)\in\Psi_2}F_{ik}F_{il}^*F_{jl}F_{jk}^*\cr
&\leq a_{\max}^4(2N^2-N)+\frac{1}{N}\sum_{(i,j,k,l)\in\Psi_2}F_{ik} F_{il}^*F_{jl}F_{jk}^*\cr
&=a_{\max}^4(2N^2-N)+\frac{4}{N}\sum^{N}_{i,j,k,l=1\atop i<j, k<l}Q_{ijkl}
\end{align*}
where $\Psi_1\triangleq \{(i,j,k,l)|i,j,k,l\in [1:N], i=j \mbox{ or } k=l\}$, $\Psi_2\triangleq \{(i,j,k,l)|i,j,k,l\in [1:N], i\neq j \mbox{ and } k\neq l\}$, and  $Q_{ijkl}\triangleq a_{ik}a_{il}a_{jk}a_{jl}\cos\left(\frac{2\pi}{\lambda}(d_{ik}-d_{il}-d_{jk}+d_{jl})\right)$.

Note that $Q_{ijkl}$'s for all $1\leq i<j\leq N$ and $1\leq k<l \leq N$ follow an identical distribution, but they are not necessarily independent of each other. Nevertheless, $\frac{4}{N^2(N-1)^2}\sum_{i,j,k,l=1\atop i<j, k<l}^NQ_{ijkl}$ strongly converges to $\E[Q_{1212}]$ as the following lemma shows, where the expectation is over uniform node distributions.
\begin{lemma}
\label{lemma:converge}
The sample mean $\frac{4}{N^2(N-1)^2}\sum^{N}_{i,j,k,l=1 \atop i<j, k<l} Q_{ijkl}$ strongly converges to $\E[Q_{1212}]$. That is,
\begin{align*}
\P\left(\lim_{N\rightarrow \infty}\frac{4}{N^2(N-1)^2}\sum_{i,j,k,l=1\atop i<j, k<l}^N Q_{ijkl} =\E[Q_{1212}]\right)=1 .
\end{align*}
\end{lemma}
The proof of the above lemma is given in Appendix \ref{appendix:lemma_converge}. Furthermore, the following lemma gives an upper bound on $\E[Q_{1212}]$, which is proved in Appendix \ref{appendix:mu}.
\begin{lemma}
\label{lemma:mu}$\E[Q_{1212}]=O(M^{-1})$.
\end{lemma}
From Lemmas~\ref{lemma:converge} and~\ref{lemma:mu}, we have $\E[\gamma^2]=O(\max\{N^2, N^3M^{-1}\})$ with high probability as $N$ tends to infinity.

Now, by using the bounds on $\E[\gamma]$ and $\E[\gamma^2]$, Theorem \ref{theorem:MIMO} is proved.
\endproof

\section{ Proof of Lemma~\ref{lemma:converge}}
\label{appendix:lemma_converge}
Let us first present a theorem on the strong convergence of the sample mean of a sequence of not necessarily independent random variables. The proof is in~\cite{Lyons:88}.
\begin{theorem} \label{thm:conv}
Let $\{X_m\}_{m=1}^{\infty}$ be a sequence of not necessarily independent complex-valued random variables, each of which follows an identical probability density function $f(x)$ such that $\E[X]=0$ and $\E[|X|^2]$ and $|X|$ are bounded. Suppose that
\begin{align}
\sum_{K\geq 1}\frac{1}{K^3}\E\left[\bigg|\sum_{m\leq K}X_m\bigg|^2\right]<\infty. \label{eqn:SSSN_cri}
\end{align}
Then, the strong law of large numbers holds for $\{X_m\}_{m=1}^{\infty}$, i.e.,
\begin{align*}
\lim_{K\rightarrow \infty}\frac{1}{K}\sum_{m\leq K}X_m =0 \mbox{~ almost surely.}
\end{align*}
\end{theorem}

Now, let us prove Lemma \ref{lemma:converge} using Theorem \ref{thm:conv}. For $w\in \mathbb{N}$, let $U_w$ and $V_w$ denote the collections of random variables given as
\begin{align*}
U_w&=\{Q_{ijkl}-\E[Q_{1212}]|1\leq i<j\leq w, ~1\leq k<l\leq w \}
\end{align*}
and
\begin{align*}
V_w&=
\begin{cases}
\emptyset & \mbox{if } w=1 \\
U_w\setminus U_{w-1} & \mbox{otherwise}
\end{cases}.
\end{align*}
Note that $|U_w|=\frac{w^2(w-1)^2}{4}$, $|V_w|=(w-1)^3$, and $\bigcup_{i=1}^wV_i=U_w$. Let $V_w^i$ for $i\in[1:(w-1)^3]$ be the $i$-th random variable in $V_w$ with an arbitrary ordering.
We construct a sequence $\{X_m\}_{m=1}^{\infty}$ of random variables as follows: for $m\in \mathbb{N}$, we let $X_m$ denote the random variable $V_{m'+1}^{m-\sum_{i=1}^{m'}|V_i|}$, where $m'$ is the integer satisfying $\sum_{i=1}^{m'}|V_i|+1\leq m \leq \sum^{m'+1}_{i=1}|V_i|$.

Let us show that $\{X_m\}_{m=1}^{\infty}$ satisfies the conditions in Theorem \ref{thm:conv}. First, it is easy to show that $\E[|X_m|^2]$ and $|X_m|$ are bounded, i.e., $\E[|X_m|^2]\leq a_{\max}^8$ and $|X_m|\leq 2a_{\max}^4$.
Next, the left-hand side term of the inequality in $(\ref{eqn:SSSN_cri})$ is written as
\begin{align*}
&\sum_{K\geq 1}\frac{1}{K^3}\E\left[\bigg|\sum_{m\leq K}X_m\bigg|^2\right]\cr&=\sum_{K\geq 1}\frac{1}{K^3}\sum_{m_1\leq K}\sum_{m_2\leq K}\E\left[X_{m_1}X_{m_2}\right].
\end{align*}
Consider two random variables $X_{m_1}=Q_{i_1j_1k_1l_1}-\E[Q_{1212}]$ and $X_{m_2}=Q_{i_2j_2k_2l_2}-\E[Q_{1212}]$ in $\{X_m\}_{m=1}^{\infty}$. If $\{i_1,j_1\}\bigcap\{i_2,j_2\}={\emptyset}$ and $\{k_1,l_1\}\bigcap\{k_2,l_2\}={\emptyset}$, $X_{m_1}$ and $X_{m_2}$ are independent of each other, and hence, $\E[X_{m_1}X_{m_2}]=0$. Otherwise, $|\E[X_{m_1}X_{m_2}]|=|\E[Q_{i_1j_1k_1l_1}Q_{i_2j_2k_2l_2}]-\E[Q_{1212}]^2|\leq 2a_{\max}^8$. Using this, let us get an upper bound on \[\frac{1}{K^3}\sum_{m_1\leq K}\sum_{m_2\leq K}\E\left[X_{m_1}X_{m_2}\right]\] for each $K\in \mathbb{N}$ as follows.
\begin{itemize}
\item $K=\frac{w^2(w-1)^2}{4}$ for some $w\in \mathbb{N}$:
In this case, $\{X_m|1\leq m \leq K\}$ is $U_w$. For each random variable $X_{m_1}$ in $U_w$, $\frac{(w-2)^2(w-3)^2}{4}$ random variables in $U_w$ are independent of $X_{m_1}$. Thus, we have
\begin{align*}
&\frac{1}{K^3}\sum_{m_1\leq K}\sum_{m_2\leq K}\E\left[X_{m_1}X_{m_2}\right]\cr
&\leq \frac{1}{K^3} K(2a_{\max}^8)\frac{w^2(w-1)^2-(w-2)^2(w-3)^2}{4}\cr
&=8a_{\max}^8\frac{w^2(w-1)^2-(w-2)^2(w-3)^2}{w^4(w-1)^4}\cr
&\leq C_1''a_{\max}^8\frac{1}{(w^2(w-1)^2/4)^{5/4}}\cr
&=C_1''a_{\max}^8\frac{1}{K^{5/4}}
\end{align*}
for some positive constant $C_1''$.
\item $\frac{w^2(w-1)^2}{4}<K<\frac{w^2(w+1)^2}{4}$ for some $w\in \mathbb{N}$:
Let $\hat{K}_1=\frac{w^2(w-1)^2}{4}$ and $\hat{K}_2=\frac{w^2(w+1)^2}{4}$. Then, we get
\begin{align*}
&\frac{1}{K^3}\sum_{m_1\leq K}\sum_{m_2\leq K}\E\left[X_{m_1}X_{m_2}\right]\cr
&\leq \frac{1}{K^3}\sum_{m_1\leq K}\sum_{m_2\leq K}|\E\left[X_{m_1}X_{m_2}\right]|\cr
&\leq \frac{1}{\hat{K}_1^3}\sum_{m_1\leq \hat{K}_2}\sum_{m_2\leq \hat{K}_2}|\E\left[X_{m_1}X_{m_2}\right]| \cr
&\leq \frac{1}{\hat{K}_1^3} \hat{K}_2(2a_{\max}^8)\frac{w^2(w+1)^2-(w-1)^2(w-2)^2}{4}\cr
&= 8a_{\max}^8\frac{w^2(w+1)^4-(w-2)^2(w-1)^2(w+1)^2}{w^4(w-1)^6}\cr
&\leq \frac{C_2''a_{\max}^8}{(w^2(w+1)^2/4)^{5/4}}\cr
&=C_2''a_{\max}^8\frac{1}{\hat{K}_2^{5/4}}\cr
&<C_2''a_{\max}^8\frac{1}{K^{5/4}}
\end{align*}
for some positive constant $C_2''$.
\end{itemize}
Let $C''\triangleq a_{\max}^8\max\{C_1'',C_2''\}$. Now we have
\begin{align*}
&\sum_{K\geq 1}\frac{1}{K^3}\E\left[\bigg|\sum_{m\leq K}X_m\bigg|^2\right]\cr
&=\sum_{K\geq 1}\frac{1}{K^3}\sum_{m_1\leq K}\sum_{m_2\leq K}\E\left[X_{m_1}X_{m_2}\right]\cr
&\leq C''\sum_{K\geq 1}\frac{1}{K^{5/4}},
\end{align*}
which is finite. Hence, from Theorem~\ref{thm:conv}
\begin{align*}
\P\left(\lim_{K\rightarrow \infty}\frac{1}{K}\sum_{m\leq K}X_m =0\right)=1,
\end{align*}
which concludes the proof of Lemma \ref{lemma:converge}. \endproof

\section{Proof of Lemma \ref{lemma:mu}}
\label{appendix:mu}
Consider two uniformly and independently distributed nodes $u$ and $v$ in $C_T$ and two uniformly and independently distributed nodes $s$ and $t$ in $C_R$. Consider a cartesian coordinate system whose origin is at the bottom left corner of $C_T$. Let $z_u=(x_u, y_u)$, $z_v=(x_v, y_v)$, $z_s=(x_s, y_s)$, and $z_t=(x_t, y_t)$ denote the coordinates of nodes $u$, $v$, $s$, and $t$, respectively. Let $\mathcal{S}(X)$ and $\mathcal{S}(X|Y)$ for random variables $X$ and $Y$ denote the support of the probability density function $f(x)$ and the support of the conditional probability density function $f(x|y)$, respectively. Let $\Gamma_1\subset \mathcal{S}(z_u,z_v)$ denote the set of $(z_u,z_v)$ such that the line through $z_u$ and $z_v$ intersects $C_R$, and let $\Gamma_2$ denote $\mathcal{S}(z_u,z_v)\setminus \Gamma_1$. Let $\theta\triangleq \angle uvs$ and let $\Delta(z_u,z_v,d_{sv})$ denote the length of $\mathcal{S}(\theta|z_u,z_v,d_{sv})$ where the length of an interval $[a,b]$ is defined as $b-a$.\footnote{Here, we follow the convention that $\angle{BAC}$ is the counterclockwise angle from $B$ to $C$ and $|\angle{BAC}|\leq \pi$.} Let $\phi\triangleq \angle vsu$ and let $\phi_1$, $\phi_2$, $\phi_3$, and $\phi_4$ denote $\angle vsu$ when $z_s$ is fixed at $(L,0), (L+D,0), (L+D,D),$ and $(L,D)$, respectively. Let $|\phi_m|\triangleq\min\{|\phi_1|,|\phi_2|,|\phi_3|,|\phi_4|\}$. See Fig. \ref{fig:notation}.

\begin{figure}[t]
 \centering
{
\psfrag{Ct}[c]{$C_T$}
\psfrag{r}[c]{$C_R$}
\psfrag{s}[c]{$s$} \psfrag{t}[c]{$t$} \psfrag{u}[c]{$u$} \psfrag{v}[c]{$v$}
\psfrag{theta}[c]{$\theta$}
\psfrag{phi}[c]{$\phi$}
\includegraphics[width=85mm]{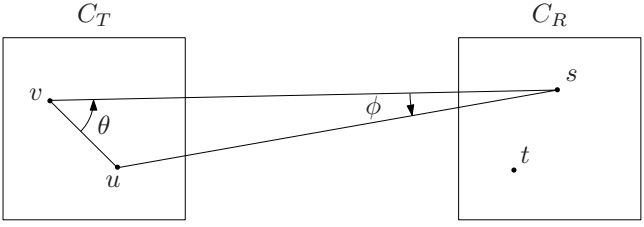}}
\caption{Two nodes $u$ and $v$ in $C_T$ and two nodes $s$ and $t$ in $C_R$. $\theta$ and $\phi$ denote $\angle uvs$ and $\angle vsu$, respectively. } \label{fig:notation}
\end{figure}

Now we are ready to prove Lemma~\ref{lemma:mu}. $\E[Q_{1212}]$ is upper-bounded as
\begin{align*}
&\E[Q_{1212}]\cr
&=\E[Q_{stuv}]\\
&=\E[\E[Q_{stuv}|z_u,z_v]]\\
&=\E[\E^2[a_{su}a_{sv}\cos(\frac{2\pi}{\lambda}(d_{su}-d_{sv}))|z_u,z_v]\cr
&~+\E^2[a_{su}a_{sv}\sin(\frac{2\pi}{\lambda}(d_{su}-d_{sv}))|z_u,z_v]]\\
&\leq \E[\E^2[a_{\max}|\E[a_{su}\cos(\frac{2\pi}{\lambda}(d_{su}-d_{sv}))|z_u,z_v,d_{sv}]||z_u,z_v]\\
&~+\E^2[a_{\max}|\E[a_{su}\sin(\frac{2\pi}{\lambda}(d_{su}-d_{sv}))|z_u,z_v,d_{sv}]||z_u,z_v]].
\end{align*}
Furthermore, $|\E[a_{su}\cos(\frac{2\pi}{\lambda}(d_{su}-d_{sv}))|z_u,z_v,d_{sv}]|$ and $|\E[a_{su}\sin(\frac{2\pi}{\lambda}(d_{su}-d_{sv}))|z_u,z_v,d_{sv}]|$ are upper-bounded as
\begin{align}
&|\E[a_{su}\cos(\frac{2\pi}{\lambda}(d_{su}-d_{sv}))|z_u,z_v,d_{sv}]|,\cr
&|\E[a_{su}\sin(\frac{2\pi}{\lambda}(d_{su}-d_{sv}))|z_u,z_v,d_{sv}]|\cr
&\leq \begin{cases}
\frac{K_{11}'a_{\max}}{\Delta(z_u,z_v,d_{sv})}\sqrt{\frac{\lambda D}{y_{uv}L}}  &\mbox{if }(z_u,z_v)\in \Gamma_1\\
\frac{K_{21}'a_{\max}\lambda}{d_{sv}\Delta(z_u,z_v,d_{sv})\sin|\phi_m|} &\mbox{if }(z_u,z_v)\in \Gamma_2
\end{cases}\label{eqn:gamma_1}
\end{align}
for some positive constants $K_{11}'$ and $K_{21}'$, where $y_{uv}\triangleq |y_u-y_v|$. These upper bounds are derived at the end of this appendix.

Using the above upper bounds, $\E[Q_{stuv}|z_u,z_v]$ is upper-bounded separately for the cases of $(z_u,z_v)
\in \Gamma_1$ and $(z_u,z_v)\in \Gamma_2$. If $(z_u,z_v)\in \Gamma_1$, we have
\begin{align*}
&\E[Q_{stuv}|z_u,z_v]\cr
&\leq \E^2[a_{\max}|\E[a_{su}\cos(\frac{2\pi}{\lambda}(d_{su}-d_{sv}))|z_u,z_v,d_{sv}]||z_u,z_v]\\
&~+\E^2[a_{\max}|\E[a_{su}\sin(\frac{2\pi}{\lambda}(d_{su}-d_{sv}))|z_u,z_v,d_{sv}]||z_u,z_v]\\
&\leq 2\E^2\left[\frac{K_{11}'a_{\max}^2}{\Delta(z_u,z_v,d_{sv})}\sqrt{\frac{\lambda D}{y_{uv}L}}
\bigg|z_u,z_v\right]\\
&\leq K_{12}'a_{\max}^4 \frac{\lambda L}{y_{uv}D}
\end{align*}
for some positive constant $K_{12}'$. If $(z_u,z_v)\in \Gamma_2$, we have
\begin{align*}
&\E[Q_{stuv}|z_u,z_v]\cr
&\leq \E^2[a_{\max}|\E[a_{su}\cos(\frac{2\pi}{\lambda}(d_{su}-d_{sv}))|z_u,z_v,d_{sv}]||z_u,z_v]\\
&~+\E^2[a_{\max}|\E[a_{su}\sin(\frac{2\pi}{\lambda}(d_{su}-d_{sv}))|z_u,z_v,d_{sv}]||z_u,z_v]\\
&\leq a_{\max}^3\E[|\E[a_{su}\cos(\frac{2\pi}{\lambda}(d_{su}-d_{sv}))|z_u,z_v,d_{sv}]||z_u,z_v]\\
&~+a_{\max}^3\E[|\E[a_{su}\sin(\frac{2\pi}{\lambda}(d_{su}-d_{sv}))|z_u,z_v,d_{sv}]||z_u,z_v]\\
&\leq 2a_{\max}^3\E\left[\frac{K_{21}'a_{\max}\lambda}{d_{sv}\Delta(z_u,z_v,d_{sv})\sin|\phi_m|}\bigg|z_u,z_v\right]\\
&\leq K_{22}'a_{\max}^4 \frac{\lambda}{D\sin|\phi_m|}
\end{align*}
for some positive constant $K_{22}'$.

Because $\E[Q_{stuv}|z_u,z_v]\leq a_{\max}^4$, $\E[Q_{stuv}|z_u,z_v]$ is upper-bounded as
\begin{align*}
&\E[Q_{stuv}|z_u,z_v]\cr
&\leq \begin{cases}
a_{\max}^4\min\left\{1, K_{12}' \frac{\lambda L}{y_{uv}D}\right\} & \mbox{if }(z_u,z_v)\in \Gamma_1,\\
a_{\max}^4\min\left\{1, K_{22}'\frac{\lambda}{D\sin|\phi_m|}\right\} &\mbox{if }(z_u,z_v)\in \Gamma_2.
\end{cases}
\end{align*}
Finally, $\E[Q_{stuv}]$ is upper-bounded as follows:
\begin{align*}
&\E[Q_{stuv}]\cr
&=\int_{\Gamma_1}\E[Q_{stuv}|z_u,z_v]f(z_u,z_v)dz_udz_v\cr
&~+\int_{\Gamma_2}\E[Q_{stuv}|z_u,z_v]f(z_u,z_v)dz_udz_v\\
&\leq a_{\max}^4 \int_{\Gamma_1}\min\left\{1, K_{12}' \frac{\lambda L}{y_{uv}D}\right\}f(z_u,z_v)dz_udz_v \cr
&~+a_{\max}^4\int_{\Gamma_2}\min\left\{1, K_{22}'\frac{\lambda}{D\sin|\phi_m|}\right\}f(z_u,z_v)dz_udz_v\\
&\leq a_{\max}^4\int_{\mathcal{S}(z_u,z_v)}\min\left\{1, K_{12}' \frac{\lambda L}{y_{uv}D}\right\}f(z_u,z_v)dz_udz_v\cr
&~+a_{\max}^4\int_{\mathcal{S}(z_u,z_v)}\min\left\{1, K_{22}'\frac{\lambda}{D\sin|\phi_m|}\right\} f(z_u,z_v)dz_udz_v\\
&\leq a_{\max}^4 \int_{\mathcal{S}(y_{uv})}\min\left\{1, K_{12}' \frac{\lambda L}{y_{uv}D}\right\}f(y_{uv})dy_{uv}\\
&~+K_{23}'a_{\max}^4\int_{\mathcal{S}(|\phi_1|)}\min\left\{1, K_{22}'\frac{\lambda}{D\sin|\phi_1|}\right\}f(|\phi_1|) d|\phi_1|\\
&\leq a_{\max}^4 \frac{\lambda L}{K_{31}' D^2}\left(1+\left(\log \frac{K_{32}'D^2}{\lambda L}\right)^+\right)
\end{align*}
for some positive constants $K_{23}'$, $K_{31}'$, and $K_{32}'$.
Since $\E[Q_{stuv}]\leq a_{\max}^4$, Lemma~\ref{lemma:mu} is proved.

%%%%
Now it remains to show the upper bounds in (\ref{eqn:gamma_1}). The upper bounds in (\ref{eqn:gamma_1}) are obtained by using the following lemma, whose proof is at the end of this appendix.
\begin{lemma} \label{lemma:int_upper}Let $g(x)$ be a periodic Lebesgue-integrable function on $\mathbb{R}$ with period $p>0$ that satisfies $g(x)=-g(x+p/2)$ and $\max_{x\in \mathbb{R}} |g(x)|=1$. Let $h(x)$ be a non-negative and Lebesgue-integrable function on $\mathbb{R}$. Consider an interval $[a,b]$ and constants $c_1\neq 0$ and $c_2$. If a partition $\Pi=\{x_0,x_1,...,x_m\}$ of $[a,b]$ exists for finite $m$ such that $a=x_0<x_1<...<x_m=b$ and $h(x)$ is monotone on each interval $[x_{i-1}, x_i]$ for $i\in[1:m]$, we have
\begin{align*}
\bigg|\int^{b}_{a}g(c_1x+c_2)h(x)dx\bigg| \leq m \int^{\tilde{x}+\frac{p}{2|c_1|}}_{\tilde{x}}h(x)dx %\label{eqn:int_upper}
\end{align*}
where $\tilde{x}\in \mathbb{R}$ is such that $\int^{\tilde{x}+\frac{p}{2|c_1|}}_{\tilde{x}}h(x)dx\geq \int^{x+\frac{p}{2|c_1|}}_{x}h(x)dx$ for all $x\in [a, b-\frac{p}{2|c_1|}]$.
\end{lemma}
\begin{figure*}\hrule
\begin{align}
|\E[a_{su}\cos(\frac{2\pi}{\lambda}(d_{su}-d_{sv}))|z_u,z_v,d_{sv}]|=\bigg|\int_{\mathcal{S}(d_{su}|z_u,z_v,d_{sv})} \cos(\frac{2\pi}{\lambda}(d_{su}-d_{sv}))a_{su}f(d_{su}|z_u,z_v,d_{sv}) dd_{su}\bigg| \label{eqn:ub}
\end{align}\hrule
\end{figure*}
To obtain an upper bound on (\ref{eqn:ub}) using Lemma~\ref{lemma:int_upper}, we first show that $\mathcal{S}(d_{su}|z_u,z_v,d_{sv})$ consists of a finite number of intervals such that $a_{su}f(d_{su}|z_u,z_v,d_{sv})$ is monotone for each. Because $|\theta|$ and $d_{su}$ have a one-to-one relationship, we have
\[a_{su}f(d_{su}|z_u,z_v,d_{sv})=a_{su}f(|\theta||z_u,z_v,d_{sv})\frac{d|\theta|}{dd_{su}}\] where \[\frac{d|\theta|}{dd_{su}}=\frac{d_{su}}{d_{sv}d_{uv}\sin |\theta|}=\frac{1}{d_{sv}\sin|\phi|}.\]
We can easily show that $a_{su}\frac{d|\theta|}{dd_{su}}$ has at most two critical points from its derivative with respect to $|\theta|$ and that $\mathcal{S}(|\theta||z_u, z_v, d_{sv})$ can be split into at most four intervals such that $f(|\theta||z_u,z_v,d_{sv})$ is a constant for each. Hence, $\mathcal{S}(|\theta||z_u, z_v, d_{sv})$ can be split into at most six intervals such that $a_{su}f(|\theta||z_u,z_v,d_{sv})\frac{d|\theta|}{dd_{su}}$ is monotone for each, implying that $\mathcal{S}(d_{su}|z_u,z_v,d_{sv})$ can also be split into at most six intervals such that $a_{su}f(d_{su}|z_u,z_v,d_{sv})$ is monotone for each. Because $a_{su}\leq a_{\max}$ and $f(|\theta||z_u,z_v,d_{sv})\leq \frac{2}{\Delta(z_u,z_v,d_{sv})}$, we have
\begin{align}
&|\E[a_{su}\cos(\frac{2\pi}{\lambda}(d_{su}-d_{sv}))|z_u,z_v,d_{sv}]|\cr
&\leq \frac{12a_{\max}}{\Delta(z_u,z_v,d_{sv})}\int_{\tilde{d}_{su}}^{\tilde{d}_{su}+\frac{\lambda}{2}} \frac{d|\theta|}{dd_{su}}dd_{su}\label{eqn:mid1}
\end{align}
from Lemma \ref{lemma:int_upper}, where $\tilde{d}_{su}$ is such that \[\int_{\tilde{d}_{su}}^{\tilde{d}_{su}+\frac{\lambda}{2}} \frac{d|\theta|}{dd_{su}}dd_{su}\geq \int_{d_{su}}^{d_{su}+\frac{\lambda}{2}} \frac{d|\theta|}{dd_{su}}dd_{su}\] for all $d_{su}\in\mathcal{S}(d_{su}|z_u,z_v,d_{sv})$.

In the same way,
\begin{align}
&|\E[a_{su}\sin(\frac{2\pi}{\lambda}(d_{su}-d_{sv}))|z_u,z_v,d_{sv}]|\cr
&\leq  \frac{12a_{\max}}{\Delta(z_u,z_v,d_{sv})}\int_{\tilde{d}_{su}}^{\tilde{d}_{su}+\frac{\lambda}{2}} \frac{d|\theta|}{dd_{su}}dd_{su}.\label{eqn:mid2}
\end{align}
We bound $\int_{\tilde{d}_{su}}^{\tilde{d}_{su}+\frac{\lambda}{2}} \frac{d|\theta|}{dd_{su}}dd_{su}$ separately for the cases of $(z_u,z_v)\in \Gamma_1$ and $(z_u,z_v)\in \Gamma_2$.  Without loss of generality, assume that $x_v\leq x_u$. First, consider the case of $(z_u, z_v)\in \Gamma_1$. Note that $ \frac{d|\theta|}{dd_{su}}$ is decreasing in $d_{su}\in[d_{sv}-d_{uv}, \sqrt{d_{sv}^2-d_{uv}^2}]$ and is increasing in $d_{su}\in[\sqrt{d_{sv}^2-d_{uv}^2}, d_{sv}+d_{uv}]$. For the case of $(z_u,z_v)\in \Gamma_1$, $d_{sv}-d_{uv}\in \mathcal{S}(d_{su}|z_u,z_v,d_{sv})\subseteq [d_{sv}-d_{uv}, \sqrt{d_{sv}^2-d_{uv}^2}]$, and hence, we have $\tilde{d}_{su}=d_{sv}-d_{uv}$. Therefore, we have
\begin{align*}
\int_{\tilde{d}_{su}}^{\tilde{d}_{su}+\frac{\lambda}{2}}  \frac{d|\theta|}{dd_{su}}dd_{su}&=\int_{d_{sv}-d_{uv}}^{d_{sv}-d_{uv}+\frac{\lambda}{2}}  \frac{d|\theta|}{dd_{su}}dd_{su}\\
&=\int_{0}^{|\hat{\theta}|} d|\theta|\\
&=|\hat{\theta}|
\end{align*}
where $|\hat{\theta}|$ is $|\angle{uvs}|$ when $d_{su}=d_{sv}-d_{uv}+\frac{\lambda}{2}$ for given $z_u,z_v,d_{sv}$. We have the following bounds on $\cos|\hat{\theta}|$:
\begin{align}
1-K_{111}'\frac{\lambda D}{y_{uv}L}\leq \cos|\hat{\theta}|\leq 1-K_{112}'|\hat{\theta}|^2 \label{eqn:cos_theta}
\end{align}
for some positive constants $K_{111}'$ and $K_{112}'$.
The upper bound holds since $|\hat{\theta}|\ll \pi$ and the lower bound is obtained as
\begin{align*}
\cos|\hat{\theta}|&=\frac{d_{sv}^2+d_{uv}^2-d_{su}^2}{2d_{sv}d_{uv}}\bigg|_{d_{su}=d_{sv}-d_{uv}+\frac{\lambda}{2}}\\
&\geq 1-\frac{\lambda}{d_{uv}}\\
&=1-\frac{\lambda}{y_{uv}}\sin|\omega|\\
&\overset{(a)}{\geq} 1-K_{111}'\frac{\lambda D}{y_{uv}L}
\end{align*}
where $\omega$ is the angle between the line through $z_u$ and $z_v$ and the horizontal line crossing $z_u$ and $(a)$ is because $(z_u,z_v)\in \Gamma_1$.
From (\ref{eqn:cos_theta}), we have \[\int_{\tilde{d}_{su}}^{\tilde{d}_{su}+\frac{\lambda}{2}}  \frac{d|\theta|}{dd_{su}}dd_{su}=|\hat{\theta}|\leq \sqrt{K_{113}'\frac{\lambda D}{y_{uv}L}}\] for some positive constant $K_{113}'$. Using this bound in (\ref{eqn:mid1}) and (\ref{eqn:mid2}), the upper bounds in (\ref{eqn:gamma_1}) for the case of $(z_u,z_v)\in \Gamma_1$ are obtained.

Now consider the case of $(z_u,z_v)\in \Gamma_2$. The following lemma gives a lower bound on $|\phi|$ for the case of $(z_u, z_v)\in \Gamma_2$, whose proof is given at the end of the present appendix.
\begin{lemma} \label{lemma:phi_lower}
When $(z_u, z_v)\in \Gamma_2$ is given, $|\phi|$ is lower-bounded by $|\phi_m|$.
\end{lemma}
From the above lemma, \[\frac{d|\theta|}{dd_{su}}=\frac{1}{d_{sv}\sin|\phi|}\leq \frac{1}{d_{sv}\sin|\phi_m|}\] and hence,
\begin{align*}
\int_{\tilde{d}_{su}}^{\tilde{d}_{su}+\frac{\lambda}{2}} \frac{d|\theta|}{dd_{su}}dd_{su}&\leq \frac{\lambda}{2d_{sv}\sin|\phi_m|}.
\end{align*}
Using this bound in (\ref{eqn:mid1}) and (\ref{eqn:mid2}), the upper bounds in  (\ref{eqn:gamma_1}) for the case of $(z_u,z_v)\in \Gamma_2$ are proved. \endproof

\subsubsection*{Proof of Lemma \ref{lemma:int_upper}}
It can be easily shown that for any interval $[a_1, b_1]$ on which $h(x)$ is monotonically increasing, we have
\begin{align}
\bigg|\int^{b_1}_{a_1}g(c_1x+c_2)h(x)dx\bigg|\leq \int^{b_1}_{b_1-\frac{p}{2|c_1|}}h(x)dx, \label{eqn:int_upper_inc}
\end{align}
and for any interval $[a_2, b_2]$ on which $h(x)$ is monotonically decreasing,
\begin{align}
\bigg|\int^{b_2}_{a_2}g(c_1x+c_2)h(x)dx\bigg|\leq \int^{a_2+\frac{p}{2|c_1|}}_{a_2}h(x)dx. \label{eqn:int_upper_dec}
\end{align}
From (\ref{eqn:int_upper_inc}) and (\ref{eqn:int_upper_dec}), Lemma \ref{lemma:int_upper} is directly obtained. \endproof

\subsubsection*{Proof of Lemma \ref{lemma:phi_lower}}
Assume that $(z_u, z_v)\in \Gamma_2$ is given. Then, $\mathcal{S}(\phi|z_u,z_v)$ is included in either $[-\pi,0)$ or $(0, \pi]$. $|\phi|$ is given as follows:
\begin{align*}
|\phi|=\arccos \frac{d_{su}^2+d_{sv}^2-d_{uv}^2}{2d_{su}d_{sv}}.
\end{align*}
Fix $x_s$. The derivative of $\frac{d_{su}^2+d_{sv}^2-d_{uv}^2}{2d_{su}d_{sv}}$ with respect to $y_s$ has the form of a rational polynomial $\frac{g_1(y_s)}{g_2(y_s)}$, where $g_2(y_s)$ is positive for every $y_s\in [0,D]$ and $g_1(y_s)$ is a cubic function of $y_s$ with a positive cubic coefficient whose roots are given as (\ref{eqn:roots}).
\begin{figure*}\hrule
\begin{align}
\frac{x_s(y_u-y_v)-x_vy_u+x_uy_v}{x_u-x_v},\frac{x_s(y_u-y_v)-x_vy_u+x_uy_v\pm \left((x_s-x_u)(x_s-x_v)((x_u-x_v)^2+(y_u-y_v)^2)\right)^{\frac{1}{2}}}{x_u-x_v}\label{eqn:roots}
\end{align}\hrule
\end{figure*}
Since $\phi=0$ when $y_s=\frac{x_s(y_u-y_v)-x_vy_u+x_uy_v}{x_u-x_v}$, which violates the assumption $(z_u, z_v)\in\Gamma_2$, $\mathcal{S}(y_s)$ contains at most one root of $g_1(y_s)$. Because the cubic coefficient of $g_1(y_s)$ is positive, $\frac{d_{su}^2+d_{sv}^2-d_{uv}^2}{2d_{su}d_{sv}}$ is maximized when $y_s$ is 0 or $D$, and hence, $|\phi|$ is minimized when $y_s$ is 0 or $D$.

In a similar way, we can show that $|\phi|$ is minimized when $x_s$ is $L$ or $L+D$ for fixed $y_s$. Thus, $|\phi|$ is lower bounded by $|\phi_m|$. \endproof

\section{Extension to a Path-loss Exponent Larger than Two} \label{appendix:path_loss_two}
In this appendix, we consider the channel model with a path-loss exponent larger than two, i.e., the discrete-time baseband-equivalent channel gain (\ref{eqn:H_su}) between nodes $k$ and $i$ at time $m$ is changed to
\begin{align}
H_{ik}[m]&=\frac{\sqrt{G}}{d_{ik}[m]^{\alpha/2}}\exp\left(-j\frac{2\pi}{\lambda}d_{ik}[m]\right)
\label{eqn:H_su_2}
\end{align}
with the path-loss exponent $\alpha>2$. For $\alpha=4$, this channel model approximates the channel when there are a direct path and a reflected path off the ground plane between transmit and receive antennas with a sufficiently large horizontal distance. For $\alpha>2$ and $\alpha\neq 4$, however, the channel model (\ref{eqn:H_su_2}) is not a direct consequence of Maxwell's equations, and hence, the DoF limit characterized in \cite{Franceschetti:09} is not valid for this channel model.

Now, let us present throughput scalings using the modified HC scheme constructed in Section \ref{sec:modifiedHC} for the channel model in (\ref{eqn:H_su_2}). In dense networks, we can get the same throughput scaling in Theorem \ref{thm:dense}. In extended networks, the throughput scaling using the modified HC scheme is decreased because the network becomes power-limited. For the same transmit power, the received power at each node in extended networks is decreased by a factor of $n^{\alpha/2-1}$ as compared to the dense network. By rescaling the space, an extended network can be considered as an equivalent dense network on a unit area but with the average power constraint per node reduced to $P/n^{\alpha/2-1}$ instead of $P$ and the wavelength reduced to $\lambda n^{-1/2}$ instead of $\lambda$. Note that the average power constraint $P/n^{\alpha/2-1}$ per node is less than $P$. As the bursty modification of the HC scheme in \cite{Ozgur:07}, we use the bursty version of the modified HC scheme, i.e., we use the modified HC scheme with operating power $P$ for $n^{1-\alpha/2}$ fraction of the time and keep silent for the remaining fraction of the time. This satisfies the average power constraint per node $P/n^{\alpha/2-1}$ and yields an aggregate throughput scaling of $n^{1-\frac{\alpha}{2}}\min\left\{\sqrt{n}\lambda^{-1},n\right\}^{1-\epsilon}$.

\bibliographystyle{IEEEtran}

\begin{IEEEbiographynophoto}{Si-Hyeon Lee} (S'08) received the B.S. degree, summa cum laude, in electrical engineering from the Korea Advanced Institute of Science and Technology (KAIST), Daejeon, South Korea, in 2007. She is currently pursuing the Ph.D. degree in electrical engineering at KAIST. Her research interests include network information theory and wireless communication systems.
\end{IEEEbiographynophoto}

\begin{IEEEbiographynophoto}{Sae-Young Chung} (S'89-M'00-SM'07) received the B.S. (summa cum laude) and M.S. degrees in electrical engineering from Seoul National University, Seoul, South Korea, in 1990 and 1992, respectively and the Ph.D. degree in electrical engineering and computer science from the Massachusetts Institute of Technology, Cambridge, MA, USA, in 2000. From September 2000 to December 2004, he was with Airvana, Inc., Chelmsford, MA, USA. Since January 2005, he has been with the Department of Electrical Engineering, Korea Advanced Institute of Science and Technology, Daejeon, South Korea, where he is currently a KAIST Chair Professor. He has served as an Associate Editor of the IEEE Transactions on Communications since 2009. He is the Technical Program Co-Chair of the 2014 IEEE International Symposium on Information Theory. His research interests include network information theory, coding theory, and wireless communications.
\end{IEEEbiographynophoto}
\end{document}